\documentclass[11pt,twoside,letter]{article}
\usepackage{comment}
\usepackage{amsthm, amsmath, amssymb, amsfonts}
\usepackage{mathtools}

\usepackage{bbm}
\usepackage{mathrsfs}
\usepackage{cancel}

\usepackage{graphicx}

\usepackage{epstopdf}

\usepackage{wrapfig}

\usepackage[shortlabels]{enumitem}

\usepackage[usenames,dvipsnames]{color}

\usepackage[colorlinks=true, pdfstartview=FitV, linkcolor=red, citecolor=blue, urlcolor=blue]{hyperref}
\usepackage[usenames]{color}

\usepackage{url}
\makeatletter\def\url@leostyle{%
 \@ifundefined{selectfont}{\def\UrlFont{\sf}}{\def\UrlFont{\scriptsize\ttfamily}}} \makeatother

\usepackage[margin=1.0in, letterpaper]{geometry}

\newtheorem{theorem}{Theorem}

\newtheorem{lemma}[theorem]{Lemma}

\theoremstyle{definition}
\newtheorem{definition}[theorem]{Definition}

\theoremstyle{remark}
\newtheorem{remark}[theorem]{Remark}

\numberwithin{equation}{section}
\numberwithin{theorem}{section}

\definecolor{Red}{rgb}{0.9,0,0.0}
\definecolor{Blue}{rgb}{0,0.0,1.0}

\def\rq{\mathtt{q}}
\def\cA{\mathcal{A}}

\def\cC{\mathcal{C}}

\def\cG{\mathcal{G}}

\def\cN{\mathcal{N}}

\def\cQ{\mathcal{Q}}
\def\cR{\mathcal{R}}
\def\cS{\mathcal{S}}
\def\cT{\mathcal{T}}
\def\cU{\mathcal{U}}
\def\cV{\mathcal{V}}
\def\cW{\mathcal{W}}
\def\cX{\mathcal{X}}

\def\bE{\mathbb{E}}
\def\bF{\mathbb{F}}

\def\bP{\mathbb{P}}

\def\bR{\mathbb{R}}

\def\bT{\mathbb{T}}

\def\sF{\mathscr{F}}

\def\sS{\mathscr{S}}

\def\bfC{\mathbf{C}}

\def\bfQ{\mathbf{Q}}
\def\bfR{\mathbf{R}}

\def\bfW{\mathbf{W}}

\def\bfb{\mathbf{b}}

\def\bfr{\mathbf{r}}

\def\bfw{\mathbf{w}}

\def\bfpi{\boldsymbol{\pi}}

\newcommand{\set}[1]{\{#1\}}            
\newcommand{\Set}[1]{\left\{#1\right\}} 
\renewcommand{\mid}{\;|\;}              

\DeclareMathOperator*{\argmin}{arg\,min} 
\DeclareMathOperator*{\argmax}{arg\,max} 

\title{Multiperiod Groundwater Markets}

\author{ 
	Igor Cialenco\,\thanks{Department of Applied Mathematics, Illinois Institute of Technology
		\newline \hspace*{1.45em}  10 W 32nd Str, Building RE, Room 220, Chicago, IL 60616, USA
		\newline \hspace*{1.45em}  Email: \url{cialenco@iit.edu},  URL: \url{http://cialenco.com}
        \newline \hspace*{1.45em} ORCID: \url{https://orcid.org/0000-0002-1825-0097}
		\vspace{0.5em}} 
	\and and \qquad  
	 Michael Ludkovski\,\thanks{Department of Statistics and Applied Probability, University of California Santa Barbara
		\newline \hspace*{1.45em}  South Hall, Santa Barbara, CA 93106-3110, USA
		\newline \hspace*{1.45em} Email: \url{ludkovski@pstat.ucsb.edu}, URL: \url{http://ludkovski.pstat.ucsb.edu/}, 
        \newline \hspace*{1.45em} ORCID: \url{https://orcid.org/0000-0001-7887-3870}
		\vspace{0.5em}}
}

\date{ {\small  Preliminary Draft: \today}} %

\begin{document}
\date{ {\small  First circulated: May 24, 2026 \\ This version: June 4, 2026}}  

	\maketitle

	\vspace{-1em}
	
	
	{\footnotesize
		\begin{tabular}{l@{} p{350pt}}
			\hline \\[-.2em]
			\textsc{Abstract}: \ & 			
            Motivated by the emergence of local groundwater exchanges, we construct and analyze stochastic models of dynamic groundwater markets. Our primary focus is endogenizing the price formation and groundwater pumping strategies in a closed market with stochastic groundwater allocations and opportunities for intertemporal transfer through rights banking. In our model, several agents, interpreted as farmers or agricultural districts, make competitive decisions on water consumption to produce a basket of goods, as well as on trading allocations among themselves, or banking them for future periods.  We define the respective discrete-time non-zero-sum non-cooperative game and construct its sub-game perfect Nash equilibria characterized by the groundwater price process $\{p^\circ(t)\}$. We furthermore construct an algorithm to determine equilibrium strategies and prices through a machine learning approach on top of best-response iterations.  Extensive numerical experiments illustrate dynamic phenomena, including the role of groundwater recharge dynamics, agents' risk aversion and groundwater allocations. Our model provides insights into competitive effects in environmental markets with banking features.	\\[0.5em]
			\textsc{Keywords:} \ &  dynamic groundwater market, groundwater price, subgame perfect Nash equilibrium, water rights, water banking, Pareto optimality.   \\[0.5em]
			\textsc{MSC2020:} \ &  91B76, 91A10, 91B72, 91B70, 91G60  \\[0.5em]
			\textsc{JEL:} \ &  C73, Q25,  D47, C63, Q58   \\[1em]

            \hline
		\end{tabular}
		
	}



\section{Introduction}

In this paper we develop a model for dynamic groundwater markets. Our work is motivated by the emerging California groundwater landscape, where the SGMA legislation \cite{sgma2014,url:SGMA} mandates the establishment of local Groundwater Sustainability Agencies (GSA) for each groundwater basin. The GSAs must create pumping regulations, which in practice require substantial reduction in pumping volume and much stricter usage. To this end, GSAs introduce groundwater allocations to their stakeholders; for economic efficiency these allocations are often tradeable. The resulting groundwater markets have been estimated to yield substantial gains from trade by moving groundwater consumption from less efficient to most (economically) productive users \cite{ayres2021environmental,bruno2020gains}. As the next step, in order to manage the substantial and intrinsic \emph{temporal} variability of groundwater supplies, there is an emerging interest in multi-period allocation management. Specifically, groundwater banking, which allows stakeholders to transfer their water rights from one period to the next, is one such tool. 

To understand the behavior of groundwater markets with banking, we build a stochastic model that captures the strategic interactions of the respective stakeholders. The essential ingredients of our setup are as follows. The agents are noncooperative and aim to optimize their risk-adjusted future expected profit over the time horizon. To do so, they have three control variables at their disposal at each period. First, they decide how to consume water, by optimally allocating it across their productive activities (think crops). Second, they may trade their pumping rights, buying or selling them from other agents. Third, they may bank some rights, transferring them to the next period. The market operates in a Pareto framework, with the market clearing price determined from supply/demand equilibrium for groundwater buyers and sellers. For that part, we postulate a strict market clearing constraint which under our assumptions yields a unique Pareto price $p^\circ$. This $p^\circ$ is the source of the coupling between the agents, driving their trading and consumption decisions; it emerges endogenously from market equilibrium at the present period. The stochasticity is injected from the random evolution of groundwater supply, captured by the water table height process $H$. Groundwater increases stochastically from precipitation and other recharge sources, captured by the stochastic state $R(t)$, and decreases  based on the aggregate water consumption by the agents. A final aspect that drives market behavior is the allocation mechanism that determines how pumping rights (in addition to those banked in the previous year) are assigned to each agent at each period. 

Qualitatively, our market model therefore features a stochastic supply and an endogenously determined (via a game equilibrium) demand, where agents optimize both within-period (how much to consume and how much to trade) and across periods (how much to bank). These features, together with the coupling through the endogenous price process $p^\circ(t)$ is what makes our setup specific to groundwater. Nevertheless, one can consider other relevant contexts where noncooperating agents are assigned allocations of a scarce resource and interact through the resulting market-clearing price, such as fisheries or forestry.

Let us summarize our contributions. 
As a first key result, in Section \ref{sec:sub-game-perfect} we establish that the described market is indeed a sub-game perfect Markov Nash Equilibrium, which allows us to apply the dynamic programming principle (DPP) to decouple the overall problem into a series of (recursive) one-period equilibria indexed by $t$ and viewed as a function of the market state. This result is essential to achieve modeling and computational tractability and opens the door for quantitative analysis of such markets. Moreover, Theorems \ref{th:main} and \ref{th:ne-market-clearing} show that the obtained NE consists of pure strategies and clears the market.

Section \ref{sec:algorithm} then develops a numerical algorithm that iteratively computes equilibrium strategies and game values through backward recursion. This entails a two-layer procedure: in the inner layer given the water allocations (conditional on banking decisions), we determine the resulting $p^\circ_t$, the trading amounts and the individual consumption/production schedules. The latter actions are decoupled conditional on $p^\circ_t$ and obtained in parallel. In the outer layer, we determine the banking decisions which requires finding a Nash equilibrium of a non-zero continuous-state continuous-action static game. To operationalize the above, we combine the methods of \cite{CialencoLudkovski2025} for the 1-period model with a machine learning approach that employs best-response iterations to find the banking equilibrium at a collection of training inputs and then stitches them together via a spline. 

In the third part of the paper, we investigate the emergent behavior and the comparative statics of the developed dynamic market. From a market design perspective, we investigate in Section~\ref{sec:statics} the value of banking, the endogenized price dynamics and the role of allocation mechanisms. From the environmental uncertainty angle, we furthermore study how the variability and the distribution of recharge flows affect agents decisions and especially banking actions. We also consider the impact of risk-aversion and stakeholder flexibility.  In line with our concrete and real-life inspired setup, Section \ref{sec:calilbration} furthermore presents a case study reflective of a potential groundwater market in a typical contemporary California GSA. The case study focuses on the interaction between forage and orchard farmers and includes calibrated groundwater dynamics and profit functions. 

\subsection{Existing Literature on Multiperiod Groundwater Markets}

Many of the most productive agricultural regions cannot sustain rain-fed harvests, making groundwater essential. For instance, 96 percent of harvested cropland in California is irrigated \cite{SearsEtAl2019} and during droughts, when surface water supplies are greatly reduced, groundwater constitutes close to two-thirds of the water used in the state \cite{sears2022moment}. To avoid the externality of unfettered access to the shared, hydraulically connected groundwater resource, groundwater allocations are a cornerstone of management. 

Groundwater trading acts as a flexible, market-based tool in spatial water management to reallocate water to higher-value uses, reducing the economic impact of scarcity and regulatory restrictions. In California, trading programs assist GSAs in meeting sustainability goals by providing a mechanism for cap-and-trade, mitigating rigid groundwater allocations.  Compared to top-down regulatory approaches, the nature of trading programs provides groundwater users with more choices and flexibility on a voluntary basis to meet groundwater management objectives. Further, the ability to monetize groundwater in a trading program can reward individuals for conservation and improved management practices \cite{babbitt2017groundwater}. 

A natural complement to groundwater trading is the mechanism of groundwater allocation \emph{banking}, also known as Carryover Rights, In-Lieu Recharge, Soft Caps, and Intertemporal Trading. Banking allows a rights holder to defer the use of an assigned water quota in one period (by forgoing pumping) and access it in a future period. Such carryovers have been implemented in Australia \cite{hughes2009management}, Kansas \cite{peck2006groundwater} (where all pumping rights are for a period of 5 years rather than a single year) and several of California's GSAs \cite{HanakEtAl2023}. 

Beyond the primary reason of added economic efficiency \cite{young2021hydrologic,AyresEtAl2021,arnold2021economic} through additional flexibility that optimizes economic yields, banking eliminates use-it-or-lose-it behavior where agents will pump their maximum entitlement to preserve their history of use. Carryover removes this perverse incentive, promoting conservation during wet/average years. Similarly, banking maximizes the economic value per acre-foot by aligning consumption with marginal value of pumped water, so that consumption is lowered during wet years when groundwater value is low and increased during droughts when it is at a premium. Thus, groundwater carryover acts as an economic stabilizer by allowing unused water allocations to be stored for future use, smoothing agricultural revenue and surface water supply volatility.  
Unlike expensive surface reservoirs,  forgoing pumping carries zero capital expenditures, monetizing the free storage capabilities of the aquifer.
 From a macroeconomic angle, multi-year allocations and carryover arrangements help farmers manage risk and maintain steady profitability. Moreover, banking  supports the sustainability of permanent crops that require consistent water access, preventing long-term capital losses.  Environmentally, groundwater banking actions  provide broad local and regional benefits, including improved water supply reliability, better drought resilience, support for river and ecosystem health, reduced water pollution, and assistance with flood control. 

At the same time, banking also creates new challenges \cite{maliva2014groundwater}. An ongoing policy knot is creating strong legal systems that track banked water, which for instance requires distinguishing between appropriative and overlying owners. Concerted ``cashing in'' of carryover rights during a drought may create a ``run on the bank'', leading to a catastrophic drawdown. Indeed, carryovers lead to more lumpy pumping; concentrating the pumping volume across time may crash hydraulic heads, dewater shallow wells, and decouple stream-aquifer interactions. Moreover, the physics of ``water in water out'' are complex. To match hydrogeologic reality, banked rights ought to decay due to natural outflows and subsurface aquifer transients.  According to Darcy's law, groundwater migrates down-gradient, flowing laterally towards the basin outlet or surface river. As such, carryover credits often  require a Leave-behind or ``haircut'' deduction to account for this migration \cite{hanak2011managing}. Similarly,  phreatophytic evapotranspiration implies that ``saved'' water in areas with shallow water table may be consumed by deep-rooted vegetation (phreatophytes) or lost to soil evaporation. Finally, hydrological conditions may create spatial interactions \cite{palazzo2014role}; for instance if neighboring pumpers lower the pressure head, clay layers may compress, so that banked groundwater is forced out to the surface as  the physical space to hold it shrinks due to compaction. As shown by \cite{palazzo2014role}, carryover rights naturally shift pumping toward the most sensitive areas of the aquifer, i.e., those with highest hydraulic conductivity, accelerating depletion.

To deal with the above issues,  regulators usually implement specific guardrails to prevent basin failure and remedy externalities imposed on third parties, such as compromised water quality and ecosystem degradation \cite{bruno2024designing}.  These include  Rollover Caps which limit how much can be banked to prevent hoarding \cite{peck2006groundwater}; Spillable Accounts if the aquifer ever fills up to the brim \cite{hughes2009management,hughes2023economic}; Cut to the Aquifer that taxes banked water \cite{hanak2011managing}, and Spatial Zoning whereby carryover credits cannot be transferred into ``red zones''. Ultimately, there is always a trade-off. For example, \cite{young2021hydrologic} show that soft cap allocation policies that allow carryover outperform equivalent hard caps in terms of greater farm profitability and reduced production risks but come at the cost of higher peak and cumulative rates of streamflow depletion.

We refer to \cite{newman2018western,bruno2021missing,bruno2024designing} for further discussion of groundwater market design and related political economies needed to achieve allocative efficiency.  For example, since most trades are bilateral, there is a strong informational asymmetry and \cite{bruno2020gains} address market power effects, such as the potential formation of water ``cartels''. \cite{bruno2024designing} urge to study how groundwater markets perform in response to a range of climate-induced water supply disruptions and  the potential of markets to reduce costs from a changing water landscape. \cite{maples2018leveraging} advocate for interdisciplinary approach to groundwater management, calling for better quantification and data clearinghouses (``you cannot manage what you do not measure'') and clearer legal frameworks.

Our setup is closest to the model in \cite{SearsEtAl2019,SearsEtAl2020} who similarly consider a dynamic game among multiple stakeholders. \cite{SearsEtAl2019} consider multiple patches of land, where agents interact due to lateral groundflows driven by table height differentials (Darcy's Law). Their state variable is the table height or stock $S$ which linearly impacts extraction costs; farmers choose the quantity $C$ of water to apply for irrigation, which enters their revenues according to a concave power function and their costs linearly. Table height increases through exogenous precipitation-driven recharge. The authors then numerically solve for the Markov Perfect Nash Equilibrium (MPNE) by iterating the best-response maps. 
Unlike our model, \cite{SearsEtAl2019}  does not allow any trading and restricts attention to a highly discretized action space (e.g.~consisting of just 4 pumping levels). \cite{SearsEtAl2020} drop the spatial aspect and instead introduce the interaction between overlying farmers and appropriators who have a linear revenue from pumped water but a limited allocation of rights (again, no trading is allowed). In a follow-up working paper \cite{sears2022moment} present an alternative approach to determine the MPNE. \cite{hughes2023economic} builds a similar dynamic program for a \emph{surface} water market in South Australia, documenting the high value of banking and its importance in reducing price and water supply volatility. Unlike our analysis, \cite{hughes2023economic} use a stylized equilibrium characterization and moreover postulate a parametric (linear-exponential) form of the value functions, rather than undertake direct optimization. The earlier \cite{hughes2009management} considered shared storage of surface water in centralized dams in Australia. 

On a broader plane, a large body of literature studies dynamic groundwater management, typically contrasting the perspective of a single planner with a competitive outcome. This strand does not incorporate markets, instead studying a range of externalities that are present for example due to common access aquifer property.  Intertemporal management naturally raises the consumption-storage decision connecting to our banking equilibrium. The common pool aspect creates two externalities: pumping externality means that withdrawals by one farmer lower the water table, increasing pumping costs of everyone else; strategic externality means that water saved by one farmer today, can be exploited by everyone else tomorrow.  For example, \cite{merrill2018optimal} consider the context of aquifer depletion which forces farmers to gradually switch to dryland farming, calibrating to the Ogalalla aquifer in Kansas. In their model, stochastic surface irrigation is likewise captured by a Markov chain and they solve the DPP to compare deadweight loss between a central planner and a myopic individual stakeholder. 

Complementary to our framing are large-scale hydro-economic optimization models such as the CALifornia Value Integrated Network
(CALVIN) \cite{pulido2004economic,khadem2018estimating}. The latter focus on top-down centralized management of surface reservoirs and groundwater basins, aiming to directly capture the ``carryover storage value'' of keeping water in reserve to mitigate drought impacts. Such models involve many dozens of decision variables (e.g., 84 decision variables in CALVIN for 26 surface reservoirs and 32 groundwater basins) and capture the multi-reservoir dependencies, so usually operate under perfect foresight assumptions over long time horizons (82 years in CALVIN). In comparison, we tackle non-cooperative optimization by each agent and the banking of exogenously allocated water rights, rather than direct consideration of stored water value. 

\begin{remark}
The above overview focuses on the groundwater context. Mathematically, our setup of a discrete-time finite-player stochastic game resembles recent works \cite{hofer2026markov} and \cite{welsh2026multi} and overlaps with approaches to games with a common resource \cite{dayanikli2025cooperation} and major-minor setups \cite{cui2024learning}. We also mention literature on price formation in mean field setups with random supply \cite{gomes2023random}.
\end{remark}

\section{Model and Dynamic Equilibrium}

We begin by introducing a general and flexible modeling framework that captures the key features of the current groundwater market and is designed to serve as a foundation for subsequent analyses and related studies. Our setup targets rural basins where the primary stakeholders are farmers and agricultural water districts, and the primary use of groundwater is for irrigation, i.e.~to grow agricultural products. Below we abstract from such details and use the generic economic terminology of agents and profit-loss functions. Throughout, we work in discrete-time to reflect the seasonality of agriculture where irrigation-driven decisions occur on a yearly basis.  The next section refines the model by discussing each component in turn.

On a filtered probability space $(\Omega, \sF, \bP, \bF= (\sF_t)_{t=0}^T)$, assume that $J$ economic agents trade within one basin groundwater rights among themselves, on a discrete time grid $t=0,1,\ldots, T$. We assume that the horizon $T$ is finite. The filtration $\bF$ is generated by an underlying set of stochastic factors, such as groundwater recharge, weather, and market fluctuations, whose specification is left to the modeler. Throughout, all processes are assumed to be $\bF$-adapted, with the dependence on $\omega\in\Omega$ usually suppressed. We use boldface to denote vectors, e.g.~$\boldsymbol{\varphi}=(\varphi_1,\ldots,\varphi_J)$. Fraktur letters  denote  sums of quantities across agents, e.g. $\mathfrak{W}=\sum_{j=1}^JW_j$.  

Central to our modeling framework is the \textit{groundwater recharge process} $\{R(t)\}$, where $R(t)$ is the amount of water (primarily from surface precipitation percolating down) entering the aquifer from time $t-1$ to $t$. We assume that the process $R(t)$  is given exogenously, and follows a Markovian structure. Each agent gets an allocation of groundwater rights $R_j(t)$ at the beginning of the period, which could be a fixed or random amount. In particular, it may depend directly or indirectly on the recharge process $R$. We assume that $R_j, \ j=1,\ldots, J$, are also Markov. See Section~\ref{sec:recharge-allocation} for more structural details about $R$ and $R_j$.

Each agent $j$  produces goods according to a profit-and-loss function $G_j:\Omega\times\cT\times[0,\infty)\to\bR$, which maps \textit{water use} $C_j\geq 0$, in acre-feet (ac-ft) to profit (in dollars) $G_j(t,C_j)$. We assume that for each $t$, the function $C_j \mapsto G_j(t,C_j)$ is increasing and concave on $\mathcal C_j(t)$. In particular, higher water use leads to higher profit. As a concave function on a compact interval, $G_j$ is continuous in $C_j$.

Agents also can trade water among themselves at the beginning of each time period $t$.
Groundwater trading is interpreted as individual groundwater users transferring their allocations, with no physically conveyance of groundwater between them. These trades are ``annual leases'' and only apply for the given period. 
Denote by $\psi_j(t)$ the amount of \textit{water traded}, in ac-ft,  by agent $j$, with  convention that  $\psi_j>0$ means selling water, and $\psi_j<0$ buying water.  Let $p(t)$ be the price  of traded water (in \$/ac-ft) at time  $t$. Both $\psi_j, p$ could be random. We follow the convention that the last trade happens at $t=T-1$,  and all remaining water is lost, with $\psi_j(T)=C_j(T)=0$. Then, the $j$-th agent P\&L is given by 
\begin{align}\label{eq:L}
L_j(t) := G_j(t,C_j(t)) + \psi_j(t) \cdot p(t), \quad t=0,1,\ldots, T. 
\end{align}
The traded water amounts must balance out via the \textit{market clearing condition} 
\begin{equation}\label{eq:mrktClearPsi}
\sum_{j=1}^J \psi_j(t) = 0. 
\end{equation}
The market-clearing condition is binding and imposed as a hard constraint, though it could be relaxed by allowing participants to opt out of trading.

Let $W_j(t)$ be the total amount of water available to the agent $j$ at the beginning of the time period $t$, that includes the water rights given by the regulator to agent $j$ at time $t$, as well as the unused water from the previous periods, but excludes the  water $C_j(t)$ needed for production of goods  or the amount of traded water $\psi_j(t)$. Hence, the random process $\{W_j(t)\}$ follows the dynamics  
\begin{equation}\label{eq:dynWj}
W_j(t+1) = W_j(t) + R_j(t+1) - C_j(t) - \psi_j(t), \quad t=0,\ldots, T-1, 
\end{equation}
with $W_j(0) = w_j$. Thus, pumping capacity  increases through new rights tied to the exogenous process $R(\cdot)$ and decreases due to consumption and trading. 
We impose the water budget constraints
\begin{equation}
    -\sum_{i\neq j} W_i(t)  \leq C_j(t) + \psi_j(t) \leq  W_j(t). \label{eq:waterBudgetConstr}
\end{equation}
The agents have the opportunity for \textit{groundwater banking}, that is, $C_j(t)+\psi_j(t)< W_j(t)$ yielding intertemporal shift in water consumption. To keep track of banking amounts, we let 
\begin{align}\label{eq:bj}
b_j(t) := W_j(t) - C_j(t) - \psi_j(t),
\end{align}
such that 
\begin{equation}\label{eq:bj-bounds}
0\leq b_j(t)\leq \mathfrak{W}(t) ,
\end{equation}
which represents how much water agent~$j$ banks in period $t$ and hence is carried over to the next period.\footnote{We note that the constraints on water banking \eqref{eq:waterBudgetConstr} can be strengthened  to $0\leq b_j(t)\leq \mathfrak{W}(t) - \sum_{i}\underline{c}_i(t)$.} Banking reduces available water today and increases groundwater available next period, reducing \eqref{eq:dynWj} to 
\begin{equation}\label{eq:dynWj-b}
W_j(t+1) = b_j(t) + R_j(t+1)  \quad t=0,\ldots, T-1.
\end{equation}
\noindent The assumption that $b_j(t)$ is non-negative is essential and implies that agents cannot borrow from their future allocations, but must first build up pumping reserves.

Each agent controls dynamically her production of goods and water rights she trades, which, in view of \eqref{eq:bj}, amounts to specifying two of the three quantities $C_j(t)$, $\psi_j(t)$, and $b_j(t)$, given $W_j(t)$. With a slight abuse of notation, we will refer to a strategy $\pi_j$ as a pair consisting of two of these quantities.  Denote by $\sS_j$ the set of all feasible controls $\pi_j$ of agent $j$, i.e.~stochastic processes satisfying constraints \eqref{eq:mrktClearPsi} and \eqref{eq:waterBudgetConstr}.  We also assume that $C_j(t) \in \cC_j(t) := [\underline{c}_j,  \overline{c}_j]$, for some fixed constants $0\leq \underline{c}_j < \overline{c}_j < \infty$. Such bounds on water consumption may arise either from physical limitations—reflecting the finiteness and bounded capacity of the aquifer—or from constraints imposed by production requirements. Additional constraints may be added, such as minimum and/or maximum water consumed or traded by each agent. Each agent maximizes a risk-reward functional $U_j(\cdot)$ of her revenue, discounted by the temporal factor $\delta$.  The functions $U_j:\bR\to\bR$  account for each agent's idiosyncratic utility, and are required to be monotone increasing and concave.

Overall, at time $t$, agent $j=1,\ldots,J$ selects a control $\pi_j(t) = \big(C_j(t), \psi_j(t)\big)$ from the feasible action set 
$$
\cA_j(\bfW):= \mathcal{C}_j(t) \times [-\sum_{i\neq j} W_j(t), {W}_j(t)],
$$ 
which is a compact in $\bR^{2}$, for every allocation $\bfW$.  More generally, agents may employ \textit{mixed} (or randomized) strategies, by choosing a probability distribution over the action set $\cA_j$.

Additionally, we  introduce a fictitious player, called the \emph{price-setter}, whose payoff at time $t$ is 
$L_{J+1}(\bfpi, p) : = p(s)  \big(\sum_{j=1}^J  \psi_j(t))\big)^2$, 
and who is choosing only $p$.
Without loss of generality, we assume that $p\in[\underline{p}, \bar p]$, for some fixed $0< \underline{p}<\overline{p} <\infty$. Indeed, under some natural growth condition on $G_j$, as proved below, for prices larger than a threshold $\bar p$, all agents will prefer to sell their water rights, hence nobody will buy, yielding the critical no-trade situation $\forall j,\ \psi_j=0$. Similarly, for small positive prices, all agents will prefer to buy. We denote by $\sS_{J+1}$ the set of adapted processes $p$ with values in $[\underline{p},\bar{p}]$, which is also the set of admissible strategy for agent $J+1$. Moreover, by analogy a mixed strategy for agent $J+1$ is a probability distribution on $[\underline{p},\bar{p}]$.

We consider the stochastic process $(\bfW(t), \bfR(t))$, $t=0,\ldots, T$, with state space $\cS=([0, \bar{h}]\times[0,\overline{r}])^{J+1}$, for some $\bar{h}, \overline{r}>0$. The upper bound $\bar{h}$ could be viewed as the largest amount of water an agent can save, which is a finite number  given the discrete time setup and assuming that the recharge process $R$ is bounded by $\overline{r}$.  A (mixed) strategy $\bfpi$ generates a probability measure $\bP^\pi$ on canonical space, with transition probabilities 
\begin{equation}\label{eq:trProb}
p((w',r') | (w,r), (c,b)) = \bP(\bfW_{t+1} = w', \bfR_{t+1} = r' \mid \bfW_t = w, \bfR_t = r, \bfC_t=c, \bfb_t = b ). 
\end{equation}
Given the feasible  strategies  $\pi_{-j}$ chosen by agents other than $j$, and given the price process $p$, the agent $j$  
maximizes her expected discounted payoff 
\[
A_j(\pi_j, \pi_{-j},p) = \bE^{\pi} \Big[\sum_{s=0}^T e^{-\delta s} U_j(L_j(s; \pi_j(s),p(s)))\Big]. 
\]
Without loss of generality, we assume all agents use the same discount rate $\delta$. 
Correspondingly, the fictitious player's  expected payoff is 
$$
A_{J+1}(p, \bfpi) := \bE^\pi \Big[\sum_{s=0}^T L_{J+1}(\bfpi(s),p(s))\Big]. 
$$

Agents play a stochastic dynamic Nash game \cite{BasarOlsder1999} without a central planner. 
The fair water rights price process $p^*(\cdot)$ is determined from a Nash equilibrium  point 
\[
(\pi^*, p^*) = ( (C_j^*(t), \psi_j^*(t)), (p^*(t)), \quad j=1,\ldots, J, \ t=0,\ldots,T, 
\]
such that no agent can gain by deviating, assuming the other agents keep their strategies unchanged.

\begin{definition}
The pair $(\bfpi^*, p^*)$  is called a Nash equilibrium (NE) if 
\begin{equation}\label{eq:Nash2}
\begin{split}
 \forall j=1,\ldots, J, \quad \forall \pi_j' \in \sS_j, \quad
&A_j(\pi_j^*, \bfpi_{-j}^*,p^*) \geq A_j(\pi_j', \bfpi_{-j}^*,p^*), \\
\forall p' \in \sS_{J+1}, \quad & A_{J+1}(p^*, \bfpi^*) \leq A_{J+1}(p', \bfpi^*).
\end{split}
\end{equation}
\end{definition}

From the general theory of stochastic games \cite{BasarZaccour2018}, one can show that there exists a NE in the class of mixed strategies. Usually, this equilibrium is not unique. In particular, as shown in \cite{CialencoLudkovski2025}, even for one period setup, for every price $p>0$ there exists  a (pure) NE.  From both a computational and an economic perspective, it is desirable to refine the problem and develop a method that yields a unique, or at least clearly computable, Nash equilibrium. We address this issue in the following section.

Next, we further refine the model considering particular cases that reflect various groundwater market designs.

\subsection{The recharge process and allocations}\label{sec:recharge-allocation}
Traditionally, the recharge process $R(t)$ is modeled using a percolation process calculated as the portion of precipitation that remains after accounting for losses such as evapotranspiration, runoff, and changes in soil moisture storage. Moreover, one accounts for lateral movement or hydroconductivity that captures how water diffuses in the aquifer, e.g. the appearance of cones of depression around pumping wells \cite{pfeiffer2012groundwater}. The above quantities are spatially heterogeneous and driven by geological rock structures.
We rather propose a reduced-form  model, where $R(t)$ is a modeled as a finite state Markov chain, say with three states low, medium, high, corresponding to drought, normal or wet year. This  Markov chain can be calibrated using available historical precipitation data and the water table height; see Section~\ref{sec:calilbration}. We note that the dynamics of the water table height $H(t)$ is given by 
$$
H(t+1) = H(t) + R(t+1) - \mathfrak{C}(t), \quad t=0,\ldots,T-1,  
$$ 
with the initial condition $H(0)$. Similar models for precipitation have appeared in \cite{merrill2018optimal,SearsEtAl2019} who respectively considered 3 states and 5 states for $R(t)$.

One justification for taking such a reduced-form structure for recharge is bounded rationality. 
There is considerable uncertainty about groundwater flows, especially by farmers who have limited  knowledge of hydrology and may be unaware of actual groundwater flows. Hence, a ``behavioral'' simplified allocation model is preferred to a detailed hydrological model \cite{pfeiffer2012groundwater}.

Each agent gets an allocation: $R_j(t) = \Theta_j(R(t))$ for some fixed maps $\Theta_j$ such that $\sum_j\Theta_j(R) = 1$ for all $R$. The motivating example are fixed allocation fractions $\Theta_j(R) = \theta_j R$ for some $\theta_j \in (0,1)$. Another example that we consider below is a senior-junior allocation,  $\Theta_1(R) = r_1 \vee (\theta_1 R)$, $\Theta_2(R) = (\theta_2 R) \wedge (R-r_1)$ that guarantees that the first farmer gets at least $r_1$ no matter what. Generally speaking, in this work, we assume that $\Theta_j$ does not depend on the strategy $\bfpi$ or price $p$.  

Another common interpretation is that the water regime $R$ controls the \emph{distribution} of water allocated. For example, in Section~\ref{sec:calilbration} we use $\log R_j(t) | R(t) \sim \cN(\theta_j \mu_{R(t)}, \sigma^2_{R(t)})$, so that total groundwater recharge is log-normal with parameters that are $R(t)$-dependent, and respective proportional allocations. 

In our setting we equate net water used with gross water pumped, directly linking water consumption $C_j(t)$ with volumetric groundwater pumping allocation $R_j(t)$. In practice, some water is lost due to runoff or evaporation, and some returns to the aquifer via deep percolation. This issue of whether allocations are based on volume pumped or on net consumption is important for measuring and enforcing real-life groundwater usage, but is effectively absorbed into some scaling constants in our framework.

Traditionally, pumping is both unmonitored and uncapped. Consequently, stakeholders are able to apply as much groundwater as they wish. Because pumping costs are linked to water table height, the latter acts as both a state variable, representing the latest ``stock'' of groundwater, as well as a parameter affecting pumping decisions. Aquifer depletion and regulations such as SGMA implement mandatory metering of irrigation wells and introduce volumetric restrictions on pumping. In the present model, we assume that these caps are binding and non-bypassable, and moreover that the aquifer is in a steady-state. Consequently, the table height is constant and the allocations are linked to the ``sustainable yield'' of the aquifer. We refer to \cite{SearsEtAl2019,SearsEtAl2020,merrill2018optimal,kuwayama2013regulation} for models that do not include pumping caps and explicitly capture pumping costs as a function of groundwater stock. 

\subsection{Profit function}\label{sec:profit-func}
First, we focus on the profit function $G_j$ corresponding to agricultural producers. Farmer $j$ has the opportunity to produce $K$ types of goods (e.g.~grow different crops), and can sell good $k$ at time $t$ for a random net profit of $\overline{f}^k_j(t,\varphi_j^k)$ dollars per $\varphi_j^k$ units of good $j$. We assume that the functions $\bar{f}_j^k:\Omega\times\cT\times \cU \to \bR$ are jointly measurable on $\Omega\otimes \cU$, where $\cU\subset [0,\infty)$ is a compact endowed with Borel $\sigma$-algebra. We also assume that $\overline{f}^k_j(t,\varphi_j^k)\varphi_j^k$ are concave in $\varphi_j^k$, for all $t$. The production of one unit of good $k$ requires $a^k(t)$ ac-ft of groundwater and costs $q^k(t)$ dollars per ac-ft to pump. Thus, the profit for farmer $j$ that produces $\boldsymbol{\varphi}_j(t) = (\varphi_j^1(t), \ldots, \varphi_j^K(t))$ units of goods becomes 
\[
F_j(t,\boldsymbol{\varphi}_j(t)) = \sum_{k=1}^K [\overline{f}^k_j(t, \varphi_j^k(t)) - q^k(t)] \varphi^k_j(t) 
\]
with corresponding consumption of $C_j(t) = \sum_{k=1}^K a^k(t)\varphi_j^k(t)$ ac-ft.  Processes $\bar{f}_j^k, a^k, q^k$ can be random, e.g. depend on the recharge process $R$ or water table $H$ capturing that during dry and hot seasons more water is needed for irrigation due to higher evaporation rate, smaller profits, or  it is more expensive to pump water when the water table is low. 
In the numerical study below, we take time-stationary $\bar f_j^k = f_j^k \cdot (\varphi_j^k)^{\alpha_j^k-1}$, with some baseline price coefficient  $f_j^k\geq 0$, and elasticity parameter $\alpha_j^k\in(0,1]$.

A farmer has a fixed and limited land to use, usually with specific parts designated for particular crops, such as trees/vines,	vegetables,	grains, corn, pasture/alfalfa. Farmers choose whether to produce particular goods, leave land fallow, or reallocate land across different goods, although the latter is usually costly and rarely implemented. All these features can be captured by the profit-per-good functions $\bar{f}_j^k$, with some additional natural production bounds
\begin{align}
	0 \leq n^k_j(t) \leq  & \ \varphi^k_j(t) \leq N_j^k(t), \qquad \text{ for all } k,t,  \label{eq:constr_Bounds} 
\end{align}
for some processes $ \mathbf{n}(t), \mathbf{N}(t)$ taken as inputs to the model. Lower bounds may result from contractual production requirements or from crops that cannot be fallowed, such as trees or vines. Upper bounds are determined by physical constraints on available land.

The resulting profit function of agent $j$ is the maximum over the $\varphi^k_j$'s,
\begin{align}\label{eq:G}
G_j(t,C_j):=\max_{\boldsymbol{\varphi}_j} 
\sum_{k=1}^K [\bar{f}^k_j(t, \varphi_j^k) -q^k(t)] \varphi^k_j = \max_{\boldsymbol{\varphi}_j} F_j(t, \boldsymbol{\varphi}_j)
\end{align}
subject to \eqref{eq:constr_Bounds} and $\sum_k a^k(t)\varphi_j^k=C_j$. Note that due to the constraints, the domain of $G_j$ is $C_j \in [\underline{c}_j(t) ,\overline{c}_j(t)]=:\cC_j(t)$, where  $\underline{c}_j(t):=\sum_{k=1}^K n_j^k a^k(t)$ and $\bar{c}_j(t):=\sum_{k=1}^K N_j^k a^k(t)$.

\begin{remark}\label{rem:Gconcave}
In view of the assumption on concavity of $\bar{f}_j^k(t,\varphi_j^k)\varphi_j^k$, the objective $F_j(t, \boldsymbol{\varphi}_j)$ is concave in $\boldsymbol{\varphi}_j$, and since the feasible set of the optimization problem \eqref{eq:G} is convex and affine in $C_j$, the function $G_j(t,C_j)$ is concave in $C_j$. Moreover, if $F_j(t, \boldsymbol{\varphi}_j)$ is strictly concave, then $G_j(C)$ is strictly concave on $[\underline{c}_j,\overline{c}_j]$.
\end{remark}

\begin{remark}
Above we view the minimum consumption needs $\underline{c}_j$ as being a hard constraint. To avoid potential infeasibility, one could take $\bar{f}_j^k(t, \varphi_j^k)$ rapidly decreasing for $\varphi_j^k(t) < n_j^k(t)$, interpreted as a steep penalty for not meeting minimum production. This penalty would represent abandoning a part of the orchard or purchasing (very expensive) surface water elsewhere, or contractual fines with vendors.   
\end{remark}

\subsection{Sub-game perfect equilibria}\label{sec:sub-game-perfect}

Generally speaking the NE is not unique \cite{hofer2026markov} and choosing an economically meaningful  equilibrium is a major task. A canonical approach is to consider sub-game perfect NE, which can be equivalently described by a backward induction procedure, tentatively in a unique way, and usually only in the class of mixed Markov strategies. Sub-game perfection can be viewed as the non-cooperative extension of the dynamic programming paradigm and enforces time-consistency properties.  We refer to comprehensive survey \cite[Chapter 6.3]{BasarZaccour2018} on sub-game perfect NEs. In what follows we will consider \textit{only  Markovian policies} $\bfpi, p$, on the state space $\cS$. Thus, the process $(\boldsymbol{W}(t), \bfR(t))$  is Markov. 

Let us consider the value of the continuation game of the agent $j$,  
\[
A_j(\pi_j, \bfpi_{-j}, p; t, \bfw, \bfr) = \bE^{\bfpi} \Big[\sum_{s=t}^T U_j\big( L_j(s)\big) \ \big| \  \boldsymbol{W}(t)= \bfw, \ \bfR(t)= \bfr \Big], \quad j=1,\ldots,J,
\]
and respectively for the price-setter 
\[
A_{J+1}(\bfpi, p; t, \bfw, \bfr) = \bE^{\bfpi} \Big[\sum_{s=t}^T p(s)  \big(\sum_{j=1}^J \psi_j(s))\big)^2  \ \big| \  \boldsymbol{W}(t) = \bfw, \ \bfR(t) = \bfr \Big]. 
\]

\begin{definition}
    A strategy $(\bfpi^*,p^*)$ is called sub-game perfect NE if for any $t=0,1,\ldots, T$, 
    \[
    A_j(\pi_j^*, \bfpi_{-j}^*,p^*; t, \bfw, \bfr) \geq     A_j(\pi_j', \bfpi_{-j}^*,p^*; t, \bfw, \bfr),  
    \]
    for any $j=1,\ldots,J$, $\pi_j'\in\sS_j$, and any state $\bfw,\bfr$, and 
    \[
    A_{J+1}(\bfpi^*, p^*; t, \bfw, \bfr) \leq     A_{J+1}(\bfpi^*, p'; t, \bfw, \bfr),   
    \]
for any $p'\in\sS_{J+1}$.     
\end{definition}

\begin{theorem}
    There exists a mixed sub-game perfect NE. 
\end{theorem}
\begin{proof}
Note that the state space $\cS$, as well as action space for each agent $\cA_j(\bfw)$ is a nonempty Borel space  endowed  with standard Borel $\sigma$-algebra in $\bR^d$. Clearly, the mapping $\bfw \to \cA_j(\bfw)$ is lower measurable, and compact-valued. In view of the assumptions above, the function $U_j(L_j(t, \bfpi,p))$ is continuous in $\bfpi,p$. Indeed, first we note that the function $\sum_{k=1}^K [\overline{f}^k_j(t, \varphi_j^k) -q^k(t)]$ is continuous in $\boldsymbol{\varphi}$. Second, the feasibility set $\cG(C_j):= \set{ \boldsymbol{\varphi}_j \in \bR^K \mid \eqref{eq:constr_Bounds}, \sum_ka^k \phi_j^k=C_j}$ for $\boldsymbol{\varphi}_j$ for problem \eqref{eq:G} is a non-empty compact in $\bR^K$ for any $C_j\in \cC_j(t)$, as being a closed subset of the the compact set in \eqref{eq:constr_Bounds}. We also note that the function $C_j\to \cG(C_J)$ is upper and lower hemicontinous. Hence, by maximum theorem the function $G_j(C_J,t)$ is continuous in $C_j$ on  $\cC_j(t)$, and thus  $U_j(L_j(t,\bfpi,p))$ is  continuous in $\bfpi,p$. Moreover, the stochastic kernel \eqref{eq:trProb} is continuous in $(c,b)$. Then, by \cite[Theorem~2]{JaskiewiczNowak2017}, a sub-game perfect NE exists. 
\end{proof}

Generally speaking, a sub-game perfect NE is constructed by  the following  standard backward procedure. 
Let us denote the  best-response continuation  value as 
\begin{equation}\label{eq:Vjw}
\cV_j( \bfpi_{-j}, p; t, \bfw, r) = \sup_{\pi\in\sS_j}  A_j(\pi, \bfpi_{-j}, p; t, \bfw, r). 
\end{equation}
Consider the strategy satisfying 
\begin{align}
    \cV_j(\bfpi_{-j}, p; T, \bfw, r)  & = U_j(G_j(T,w_j\wedge \overline{c}_j)), \quad j=1,\ldots, J \\
    V_{J+1}(\bfpi; T, w, r) & = 0,
\end{align}
and for $t=T-1,\ldots, 0$, let 
\begin{align}
\cV_j(\bfpi_{-j}, p; & \, t, \bfw, r)  = \max_{C_j, \psi_j}\Big[U_j \bigl( G_j(t, C_j) + \psi_j p \bigr) \label{eq:DPP-NE}\\
& + e^{-\delta}\bE^{\bfpi}[ \cV_j({\widehat{\bfpi}}_{-j}, \widehat{p}; t+1, \bfw + \bfR(t+1) - \bfC - \boldsymbol{\psi}, R(t+1) ) \mid \bfW(t) = \bfw, R(t) = r] \Big] \nonumber  \\
\cV_{J+1}(\bfpi; & \,t, \bfw, r) =  \min_{p\in[\underline{p}, \bar p]} |\sum_{j=1}^J \psi_j|,   \label{eq:DPP-NE2}
\end{align}
where $\bfR(t) = (R_j(t))_{j\in J}$ and $\bfC = (C_1,\ldots, C_J)$, and $(\widehat\bfpi, \widehat p)$ is time $t+1$ optimal strategy.

The next two main results show that  \eqref{eq:DPP-NE}-\eqref{eq:DPP-NE2} is well defined, yields  a subgame-perfect pure Nash equilibrium, and under mild growth conditions on $U_j$ adn $G_j$, the obtained NE also clears the  market. 

\begin{theorem}\label{th:main}
    The strategy $(\bfpi^\circ, p^\circ)$ defined by \eqref{eq:DPP-NE}-\eqref{eq:DPP-NE2} is a sub-game perfect pure NE.  
\end{theorem}

\begin{theorem}\label{th:ne-market-clearing}
	Assume that $U,G_j$ are increasing, concave, and $U\in C^1(\bR), G_j(t,\cdot)\in C^1(\bR_+)$ with bounded derivatives on the admissible sets, that is there exist $M_U, M_G$ such that
	\[
	0\leq U'(x)\leq M_U, \quad 0\leq G_j'(t,C)\leq M_G,
	\]
	for all $t$ and all feasible $x,c$.\footnote{Note that in our case, bounded derivative can be reduced to continuous derivative, since given discrete time,  finite chain, and bounded constraints on $C$, the feasible set is a compact.}    
		Then the sub-game perfect equilibrium $(\bfpi^\circ,p^\circ)$ clears the market, i.e., satisfies
		\(
		\sum_{j=1}^J\psi_j^\circ(t)=0. 
		\)
	\end{theorem}
	
The proofs of Theorems~\ref{th:main} and \ref{th:ne-market-clearing} are given in the Appendix~\ref{sec:apend1}.

In what follows, we focus on the Pareto optimal (PO) equilibrium defined above, denoting all corresponding quantities as $p^\circ, b^\circ, C^\circ$, etc. The respective farmers' game values are denoted as $V_j^\circ$, being the counterpart of \eqref{eq:DPP-NE} for the Pareto case.  The key lemma is that \cite{CialencoLudkovski2025} ensures that there is a unique $p^\circ(\bfw)$
for any feasible allocation $\bfw$ allowing to iterate the recursion for as long as desired.

\subsection{Dynamic Programming}

Fixing a horizon $T$, we may view the agents as determining at each period $t$ the two controls $\pi_j=(C^\circ_j(t),b^\circ_j(t))$ ---how much to consume and how much to bank. The game aspect implies that the stochastic state is the present allocation vector $\bfw = \bfW(t)$ and the recharge regime $R(t)$. The allocation $\bfw$  determines the endogenous price $p^\circ(\bfw)$ and the corresponding consumption quantities $C^\circ_j(\bfw)$.  Conceptually, $\bfW(t)$ consists of the carry-forward banked balances $\bfb^\circ(t-1)$ and the new allocation $\Theta(R(t))$ based on the latest recharge $R(t)$, which in turn affects expectations of future recharges.

Denote by 
\begin{align}\label{eq:one-period-v}
v_j(\tilde\bfw) := U_j \left(G_j(t,C^\circ_j (\tilde\bfw) ) + \psi^\circ_j(\tilde\bfw) p^\circ(\tilde\bfw)\right) = U_j( L^\circ_j(t) ),
\end{align} 
the one-period utility of agent $j$ under the PO equilibrium and with available (non-banked) rights vector $\tilde\bfw$. 
Using dynamic programming, the effect of $b_j$ is then captured by adding the profit today that is measured by $v_j$ (and reduced through $b_j$) and the expected future profit from tomorrow until $T$ that is measured by $V^\circ_j$ (and increased through $b_j$). The present profits depend on the joint $\bfw-\bfb(t)$ and the future profits depend on the joint $\bfR({t+1})+\bfb(t)$ via \eqref{eq:dynWj-b}, coupling the players and necessitating solving for an equilibrium. Namely, to find the Pareto Nash equilibrium at period $t$ we seek 
a \emph{fixed point} $\{b^\circ_j(t), j=1,\ldots, J\}$ such that $b^\circ_j(t)$ is the best response given $\bfb^\circ_{-j}(t)$, with the conditioning on current $\bfW(t)$ and $R(t)$ done in the background.

The resulting multi-period game value $V^\circ_\cdot : \{1, \ldots, J\} \times \bT \times \bR^2_+ \times \{1, \ldots, M\} \mapsto \bR$ satisfies the DPP 
\begin{multline}
      V^\circ_j(t,\bfw, r):=   v_j \left(\bfw - \bfb^\circ(t,\bfw,r) \right) \\ + e^{-\delta} \cdot \bE \left[ V^\circ_j \left(t+1, 
    \bfR(t+1) + \bfb^\circ(t,\bfw,r), R(t+1) \right) \Big| \ R(t) = r\right],
\end{multline}
where $e^{-\delta}, \delta  \ge 0$ is the annual discount factor.

Crucially, below we break the overall equilibrium construction into two steps: given $\bfb(t)$ the price and production quantities in period $t$ are determined through the 1-period game, while $b_j^\circ(t)$ can be obtained through the fixed point game where the agents only need to pick their banking amounts. The resulting period-$t$ water price then emerges endogenously $p^\circ(t, \bfw - \bfb^\circ(\bfw, r))$. Lastly, the water consumptions $C^\circ_j(t)$ and the trading amounts $\psi^\circ_j(t)$ are determined from $p^\circ(t)$ and $\bfw$.

\subsection{Two Agents}
For two agents $J=2$, equilibrium construction
 can be reduced to examining the best response curves $\bar{b}_2 \mapsto B_1( \bar{b}_2)$ and $\bar{b}_1 \mapsto B_2( \bar{b}_1)$ and looking for their crossing points. 
Specifically, given current allocation $\bfW(t) = \bfw$ and recharge state $R(t)=r$, for any $\bar{b}_1,\bar{b}_2 \in [0, \mathfrak{W}-\sum_i \underline{c}_i(t)]$ (cf.~\eqref{eq:bj-bounds}) define
\begin{align} \label{eq:farmer1-2period}
  B_1(t,\bar{b}_2, \bfw ,r):=  &\arg\sup_{b_1 \ge 0} v_1 \!\Big(\bfw - \begin{bmatrix} b_1 
    \\ \bar{b}_2 \end{bmatrix} \!\Big) + e^{-\delta}\bE \! \left[V^\circ_1 \!\left(\! t+1, \bfR(t+1) +
    \begin{bmatrix} b_1 \\ \bar{b}_2 \end{bmatrix}, R(t+1) \!\right) \Big| \ R(t) = r\right] \!; \\ \label{eq:farmer2-2period}
  B_2(t,\bar{b}_1, \bfw, r) :=      &\arg\sup_{b_2 \ge 0} v_2 \!\left(\bfw - \begin{bmatrix} \bar{b}_1 \\ b_2 \end{bmatrix}\!\right) + e^{-\delta}\bE \left[ V^\circ_2 \left(t+1, \bfR(t+1) +
    \begin{bmatrix} \bar{b}_1 \\ b_2 \end{bmatrix}, R(t+1) \!\right) \Big| \ R(t) = r\right]. 
\end{align} 
Equilibrium $(b_1^\circ, b_2^\circ)$ is characterized by $b^\circ_j(t) = B_j(t, b_{-j}^\circ, \bfw ,r), j=1,2$. Indeed, if at any $t, \bfw, r$, we have $b_j^\circ(t) \neq B_j(t, b_{-j}^\circ,t)$ then \eqref{eq:DPP-NE} will not be satisfied as a farmer can do better. Conversely, if the fixed point $b^\circ_1(t), b^\circ_2(t)$ exists, then it yields a feasible action at that state in period $t$.

Since $b_j \in [0, \mathfrak{W}(t)]$, the optimization problems in \eqref{eq:farmer1-2period}-\eqref{eq:farmer2-2period} are of a scalar function on a bounded interval, so a solution always exists. However, equilibrium may not be interior and corner solutions of $b^\circ_j=0$ or $b^\circ_j = \mathfrak{W}(t)$ (buy up everything available and bank it all) are possible.  We expect the above best response curves to be monotone, hence there is a unique fixed point. A standard mechanism to determine the fixed point is to iterate the alternating updates $b_1^{(k)} := B_1(t, b_2^{(k-1)}, \bfw, r), b_2^{(k)} := B_2(t, b_1^{(k)}, \bfw, r)$ that are initialized from an arbitrary $b_1^{(0)}, b_2^{(0)}$ and repeated for $k=1,\ldots, K_{b}$.

\section{Numerical Algorithm}\label{sec:algorithm}

To solve \eqref{eq:farmer1-2period}-\eqref{eq:farmer2-2period} we employ nested optimization, decomposing the decisions into a 2-stage procedure. The ``inner'' decision is to pick $C_j,\psi_j$ conditional on the total water that would be consumed in the present period. This implies the respective price $p^\circ_t$ and the respective period-$t$ profits $v_j$. The ``outer'' decision is how much water to bank which is driven by the expectations of next-period value functions $V^\circ_j(t+1,\cdot)$.

The inner problem is equivalent to the 1-period setup studied in \cite{CialencoLudkovski2025}. It can be solved in closed-form conditional on $p$, the latter is determined by inverting an algebraic function that matches total desired consumption and total available (ex-post banking) groundwater. A key simplification is that for the inner problem there is no direct competition, and agent decisions can be solved in parallel, stacking up their desired consumption amounts $C_j(p)$ as a function of $p$. 

The outer problem is a ``true'' game where we must determine a competitive Nash equilibrium $\bfb^\circ$. In the outer stage, the agents must decide how much water to bank, which directly influences how much water is available in this period, as well as in the next period. This creates both a feedback effect for the agent who is banking (less profit today vs more profit next year), as well as externality for the other agent (higher prices today vs lower prices next year). Consequently, this requires determining a fixed point that we characterize through the best-response functions \eqref{eq:farmer1-2period}-\eqref{eq:farmer2-2period}. In turn $B_j$'s include the 1-period profits $v_j$ as one of the terms, coupling with the solution of the inner problem.

In the algorithm below, the state variables are the current period $t$, the current water allocation $\bfw$ and the current precipitation state $R(t)$ (hence a function in 2 continuous dimensions of $\bfw$, indexed sequentially by $t$ and indexed discretely by $r \in \{1,\ldots,M\}$). The surrogates $\widehat{V}_j(\cdot)$ are defined recursively before being passed in subsequent steps. 
\begin{enumerate}
    \item Create a subroutine that determines, for a given groundwater allocation $\bfw$, the one-period profits $v_j(\bfw)$'s by evaluating the aggregate water demand $\mathfrak{C}(p)$ as a function of $p$, and then finding through numerical one-dimensional optimization $p^\circ_1(\bfw) = \mathfrak{C}^{-1}(\bfw_1 + \bfw_2)$. 

    \item Initialize $\widehat{V}_j(T-1,\cdot, r) = v_j(\cdot)$ for all $r=r_m, m=1,\ldots, M$ (number of states of $R(\cdot)$);

    \item For $t=T-2, \ldots, 0$:
    
    \item Select a collection of $N$  groundwater endowments $\bfw^n, n=1,\ldots, N$ which act as training inputs; fix current recharge regime $r$.

    \item Define the best-response function $\widehat{B}_j(t, \bar{b}_{-j}, \bfw, r)$ for agent $j$ in period $t$, given the other's banking amount $b_{-j}$. $\widehat{B}_j$ is the counterpart of \eqref{eq:farmer1-2period}, substituting the unknown $V^\circ_j$ with $\widehat{V}_j$.
   Given $R(t) =r_k $, the conditional expectation in \eqref{eq:farmer1-2period} reduces to a  weighted sum  
   \begin{align}\label{eq:quantized-expectation}
       \sum_{m=1}^M \bfQ_{r,m} \widehat{V}_j \left(t+1, \Theta(r_m)+[ b_j, \bar{b}_{-j}]^\top, r_m\right).
   \end{align}

   \item For each $\bfw^n,r$, $n=1,\ldots, N$ find the fixed point banking amounts $\bfb^\circ(t, \bfw^n, r)$, the resulting equilibrium groundwater price $p^\circ(t, \bfw^n,r)$ and the consumption amounts $C^\circ_j(t)$,  by iterating until convergence the best-response optimization of $\widehat{B}_j$'s defined in Step 5. Save the resulting pointwise game values $$v_j^n := v_j(\bfw^n - \bfb^\circ(t,\bfw^n, r)) +  \sum_m \bfQ_{r,m} \widehat{V}_j \left(t+1, \Theta(r_m)+\bfb^\circ(t, \bfw^n, r), r_m\right)$$

  \item Train two machine learning surrogates $\widehat{V}_j(t,\cdot,r), j=1,2$ based on the $N$ pairs $(\bfw^{1:N}, v^{1:N}_j)$. 
  
  \item Repeat Steps 4-7 for each recharge state $r=r_m,m=1,\ldots,M$.

\end{enumerate}

  Note that if it takes C steps to find each $b_j^\circ(t, \bfw^n, r)$, then the algorithm entails $C \cdot M$ evaluations of $\widehat{V}^\circ(t+1, \cdot)$'s and $C$ evaluations of $v_j$'s (which requires determining $p^\circ(t,\cdot)$) in Step 1. Hence, the above algorithm leads to a total of $T \cdot N \cdot \bar{C} \cdot M$ surrogate evaluations and comparable number of $p^\circ$ evaluations, where $\bar{C}$ is the average number of steps to determine the banking equilibrium.

Specifically, we approximate $\widehat{V}_j(t,\cdot)$ using bivariate smoothing splines available in \texttt{scipy}. To train those splines,  we grid out a rectangle of $w^{n}_1, w^{n'}_2$ over $64\times 64$ allocations in period $t$, and for each grid point determine the optimal banking pair $b^\circ_j(w^n_1, w^{n'}_2)$ by doing $\bar{C}/2=12$ best-response iterations. This corresponds to solving 20 one-dimensional optimization problems as in \eqref{eq:farmer1-2period}, where the objective function is given in terms of next-period spline objects $\widehat{V}_j(t+1, \cdot)$. Total effort is therefore $64 \times 64 \times 24 = 98,304$ optimization problems to solve at each time-step which can be parallelized if multiple computing cores are available.

\begin{remark}
 while the realized available water rights $\bfW(t)$ are tightly concentrated along the diagonal (cf.~Figure \ref{fig:joint}), we have to train $\widehat{V}_j$ over a much larger grid because the numerical optimizers might consider banking amounts that are far from optimal (e.g. a very large $b_j$ while the equilibrium $b_j \simeq 0$) and hence we need a reliable estimate of the resulting future profits due to potential sub-optimal actions to avoid finding spurious equilibria. This is the key difference of our ``direct'' method vis-a-vis a Reinforcement Learning strategy that would exploit more heavily the empirical distribution of $\bfW(t)$ along equilibrium strategies.
\end{remark}

\subsection{Synthetic Case Study}\label{sec:illustrate}
 
To illustrate, we consider two agents and two crops with parameters 
\begin{align*}
\text{agent } j=1:& \qquad \alpha_1^1=0.75, \, \alpha_1^2 = 0.85, \quad f_1^1=5, f_1^2 = 9, \quad 
c_1^1= 1.5, \, c_1^2 = 3, \quad a_1^1= 1, a_1^2 = 2,\\ 
\text{agent } j=2:& \qquad  \alpha_2^1=0.75, \, \alpha_2^2 = 0.75, \quad f_2^1=8, f_2^2 = 11, \quad
c_2^1= 2, \, c_2^2 = 4, \quad a_2^1= 1, a_2^2 = 2.
\end{align*} 
Thus, Agent 1 has lower production costs for both crops $c_1^k < c_2^k$ but also lower revenue $f_1^k < f_2^k, k=1,2$, so Agent 2 is a ``premium'' high-cost farmer. Moreover because $\alpha_2^2 < \alpha_1^2$, the per-unit  revenue of Agent 2 falls faster, i.e.~the premium producer has faster-diminishing returns to scale (e.g.~due to less suitable land for their premium produce). As a result, as shown in Figure \ref{fig:Gj}, the PnL function $G_2(C_2)$ is more concave than $G_1(C_1)$. For the production feasible range we take 
$n_1^1 =  n_2^1 = 0, n_1^2 = n_2^2 = 5$ and 
$N_1^1 = N_2^1 = 50$, $N_1^2 = N_2^2 = 100$. 

For both agents,  we use log utility
\begin{align}
U_j(\ell) = \tilde{u}(\ell)= \log( \ell-20), \quad j=1,2.
\end{align}

\begin{figure}[!htb]
\centering
\includegraphics[width=0.65\textwidth]{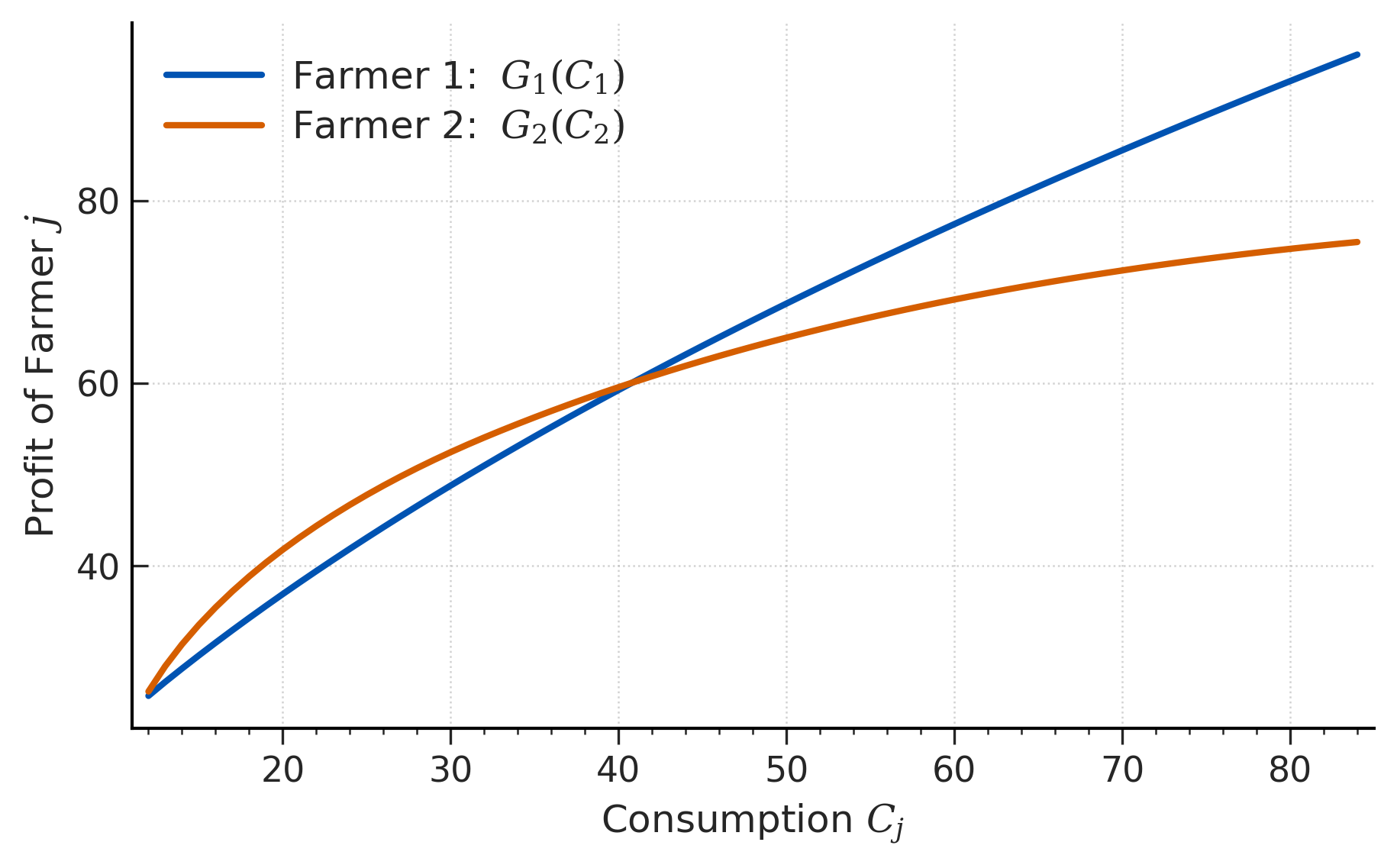}
\caption{{Profit $G_j(C_j)$} of each farmer for the base example in Section \ref{sec:illustrate}.  \label{fig:Gj}}
\end{figure}

We assume that there are $M=3$ water-year regimes with recharge amounts $R(t) \in \{40, 75, 95\}=\{r_1,r_2, r_3\}$, thought of as dry, medium, and wet years, with time-stationary dynamics $\bP(R(t+1) = r_j | R(t)=r_i) = Q_{ij}$ specified by the transition matrix $\bfQ_1$ below:
\begin{align*}
\mathbf{Q}_1 =  \begin{bmatrix} 1/3 & 1/3 & 1/3 \\
                        1/12 & 11/24 & 11/24 \\
                        1/12 & 11/24 & 11/24 \end{bmatrix}.
\end{align*}
The stationary distribution of $R$ is $\pi_R = [1/9, 4/9, 4/9]$ so that droughts occur on average every 9 years and the rest of the time recharge is equally likely to be either 75 or 95. Consequently, the unconditional mean recharge is $\bE[ R(t)] = 80$.
The allocations for the agents are proportional to $R(t)$ and specified by $\theta_1(t) = 0.6, \theta_2(t) = 0.4 \ \forall t$.

Figure \ref{fig:20-period} illustrates one scenario $\omega$ of the market across 25 periods $t=1,\ldots, 25$, visualizing all the pathwise quantities associated with a given path $R(t; \omega)$. The top panel shows the resulting groundwater price $p^\circ(t; \omega)$; the middle panels show the banking amounts $b^\circ_j(t; \omega)$, as well as the trading amounts $\psi^\circ_j(t; \omega)$. The bottom panel shows the corresponding groundwater pumping $C^\circ_j(t; \omega)$. To help with the interpretation, the dashed lines in the bottom panel visualize the recharge amounts $R_j(t; \omega)$.  Note that we initialize with additional ``reserves'' $\mathfrak{W}(0)=90$.  

The most salient impact of banking is smoothing of consumption, visualized in the bottom panel of Figure \ref{fig:20-period}. Compare the highly volatile recharge time-series (the dotted curves $R_j(t)$) against the much smoother consumption trajectory $C^\circ_j(t)$. In fact for Farmer 2, $C^\circ_2(j)$ is almost constant between $t=5$ and $t=23$. By banking a reserve of carryover pumping rights, farmers are able to withstand drought (see e.g., $t=6$) maintaining steady  consumption and hence profits.  Drought does substantially deplete banked allocations (see the drop in $b^\circ_j(6)-b^\circ_j(5)$ and $b^\circ_j(13)$ being zero), so that after multiple dry years, consumption would eventually be affected. The three consecutive years of drought at $t=11, 12, 13$ lead to zero reserves left after the third year.

In order to create this reserve cushion, farmers quickly bank up a significant reserve, witness the rapid growth in $b^\circ_j$ for $t \le 5$. Note that according to our bookkeeping notation, $b^\circ_j(t)$ represents any water that is not used up in the present; thus the case $b^\circ_j(t-1)=b^\circ_j(t)$ corresponds to the situation where the farmer had a carryover that was left as-is and banked through for the next period. In the meantime the farmer simply consumed the latest ``new'' allocation of $\Theta_j(R(t))$. So relative to $\Theta_j(R(t))$, true saving occurs when $b^\circ_j(t)>b^\circ_j(t-1)$ and conversely when $b^\circ_j(t)<b^\circ_j(t+1)$ the farmer is pumping more than allocated by using up some of their reserves. 

The top panel of Figure \ref{fig:20-period} showcases the expected negative correlation between $p^\circ(t)$ and $R(t)$: prices decrease during wet years with lots of recharge and increase during droughts. Towards the end of the simulation, the presence of the finite horizon $T$ removes the intertemporal substitution motive---farmers begin spending down their reserves in anticipation of the zero terminal condition and groundwater prices collapse until the drought at $t=24$. 
Trading between the two farmers is substantially steady. Notably, Farmer 1 buys rights from Farmer 2 under normal conditions, but sells them water in droughts (see $t=7, 12, 13$). In fact, while normally $C^\circ_1(t) > C^\circ_2(t)$ in line with higher allocation to Farmer 1, we observe that $C^\circ_2(13)>C^\circ_1(13)$ so ultimately it is Farmer 1 whose consumption is reduced the most due to a multi-year drought.

\begin{figure}[!htb]
\centering
\includegraphics[width=0.9\textwidth]{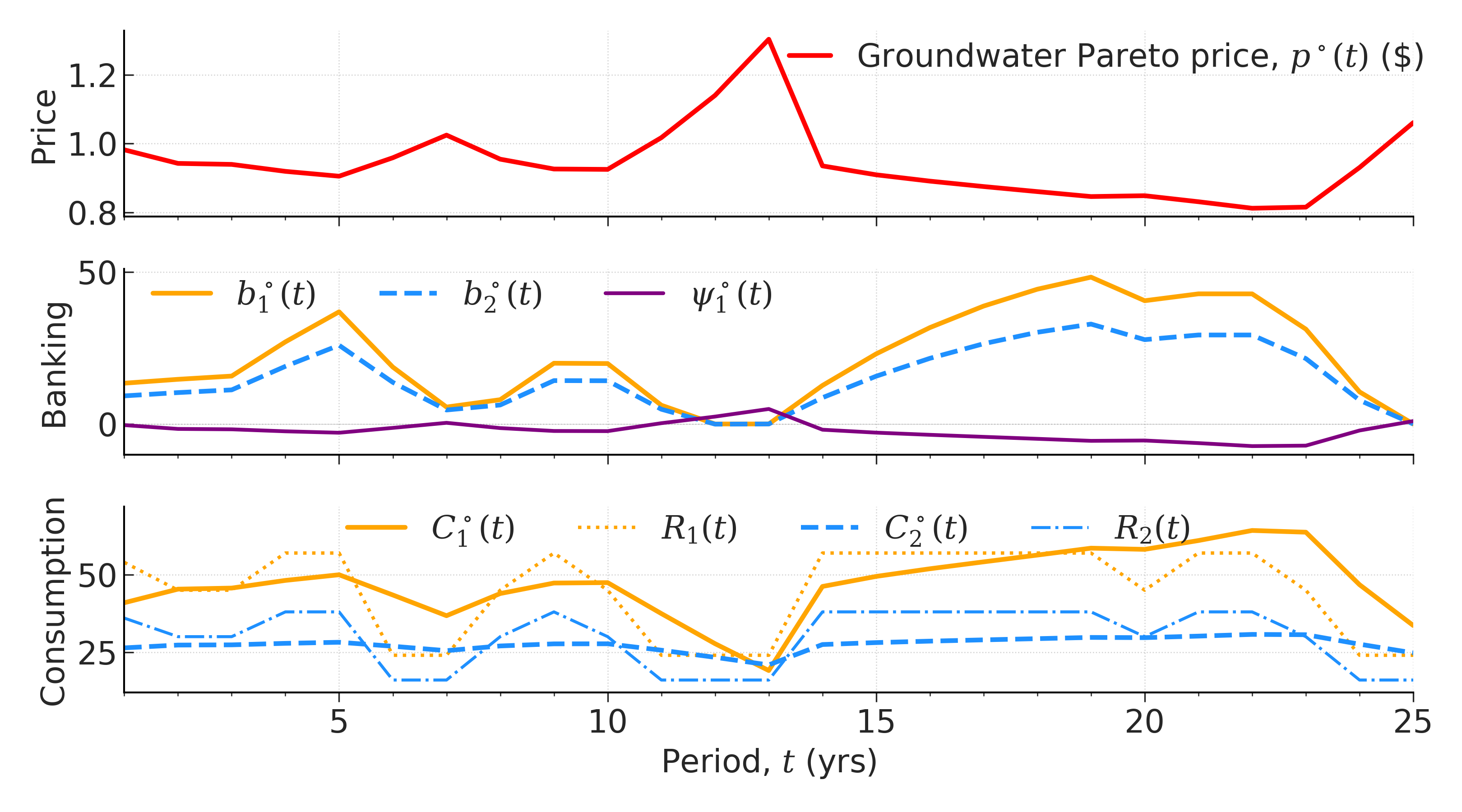}
\caption{Sample scenario of the market over $T=25$ periods. \emph{Top panel:} equilibrium groundwater price $p^\circ(t)$; Middle panel: equilibrium banking amounts $b^\circ_j(t)$ and trading amount $\psi^\circ_1(t)$; $\psi^\circ_1(t)>0$ $(\psi^\circ_1(t)<0$) indicates that Farmer 1 is buying (selling) water from Farmer 2. \emph{Bottom:} allocations of new pumping rights $R_j(t)$ and resulting equilibrium consumption amounts $C^\circ_j(t)$. By construction $C^\circ_j(t)+\psi^\circ_j(t)+b^\circ_j(t)=b^\circ_j(t-1)+R_j(t)$.  \label{fig:20-period}}
\end{figure}

To understand aggregate behavior,
Figure \ref{fig:Banking-price} shows the distributions of all the state variables and controls across 1000 randomly sampled trajectories of $R(\cdot)$.
We plot the resulting equilibrium banking amounts $b^\circ_1(t), b^\circ_2(t) \ge 0$, groundwater consumption $C^\circ_1(t), C^\circ_2(t)$, trades $\psi^\circ_1(t)$ (recall that $\psi^\circ_2(t) = -\psi^\circ_1(t)$ due to \eqref{eq:mrktClearPsi}) and groundwater prices $p^\circ(t)$. 
The plots illustrate the time-dependency inherent in our finite-horizon setup. First,  close to $T$, the risk of future droughts diminishes and farmers become more myopic, banking less and consuming more. In particular, in the last period 25 there is no banking and consumption often shoots up as farmers spend all their carry-forward allocations. The precautionary motive of banking pumping rights for a dry year starts to diminish about 4 periods prior to the ``end of time'', causing less banking to be done in years $t>20$. Second, because the initial condition is taken as fixed, for small $t$ all the distributions are discrete due to the finite state of $R(\cdot)$. At $t=0$ there is just a single realization; similarly there are only three possible values for $p^\circ(1)$ and $B_j(1)$, nine values for $p^\circ(2)$ and $b^\circ_j(2)$, etc. Third, for $t \in \{8,\ldots, 20\}$ we observe a steady-state behavior which effectively approximates the infinite-horizon setting where strategies are time-stationary.

Turning our attention to the price behavior,  the distribution of $p^\circ(t)$ has a strong right skew, meaning that equilibrium prices tend to be fairly steady in the range $[0.8, 1]$ but occasionally can be much higher. Notably, $p^\circ(t)$ is bounded above: the highest price $\bar{p} = 1.36$ is obtained during droughts when the farmers have zero carry-forward and their total consumption is just $r_1=45$.  In addition, in some scenarios we see a price collapse as we approach $T$, with prices dropping below 70 cents as farmers race to spend all their unused banked allocations due to the zero terminal condition. 

\begin{figure}[!htb]
\centering
\includegraphics[width=0.95\textwidth]{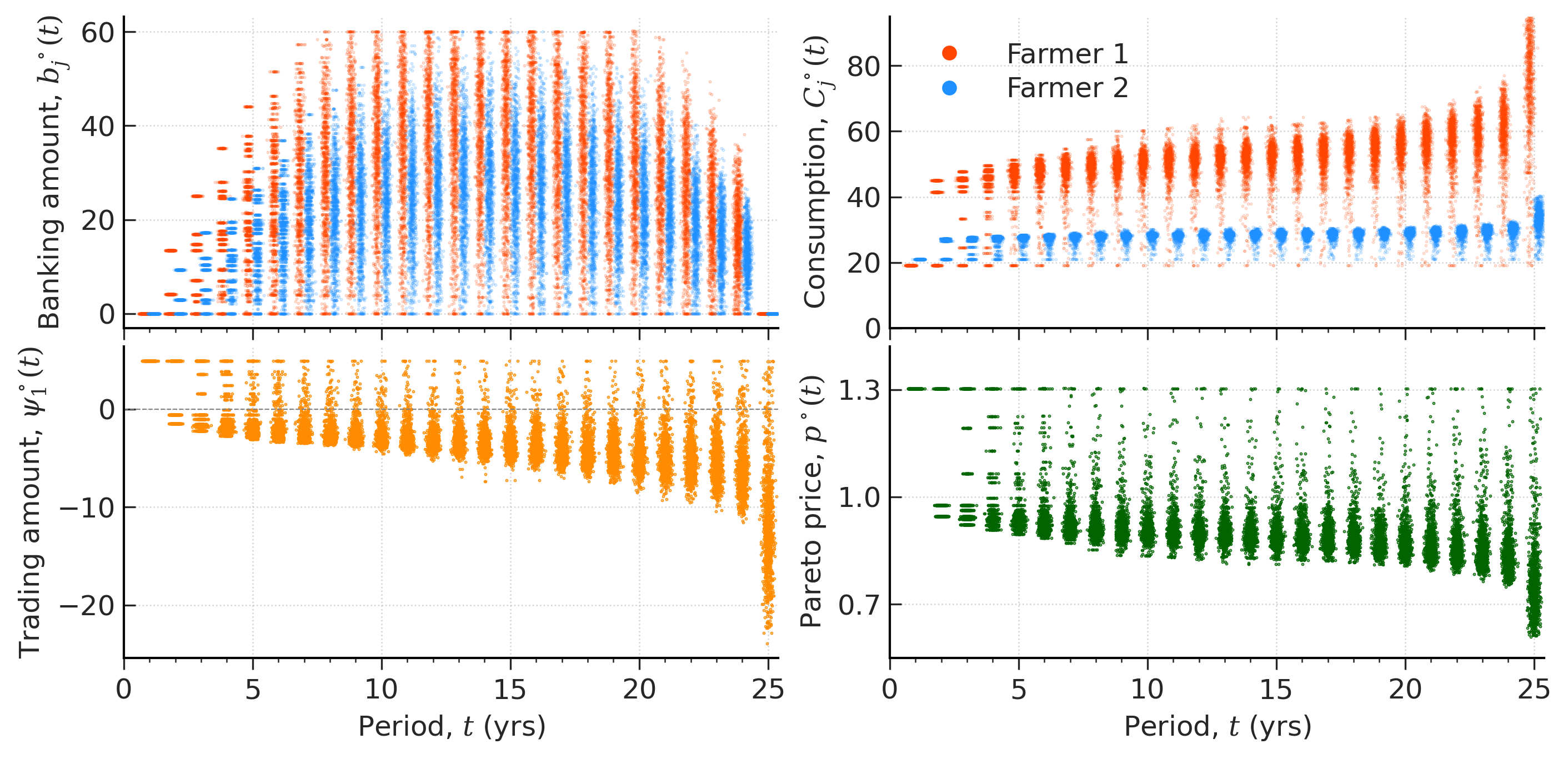}
\caption{Empirical distributions of modeled quantities across the two farmers for the case study of Section \ref{sec:illustrate}. \emph{Top Left:} equilibrium banking amounts $b^\circ_j(t)$; \emph{Top Right:} consumption quantities $C^\circ_j(t)$; \emph{Bottom Right:}
groundwater Pareto equilibrium price $p^\circ(t)$; \emph{Bottom Left:} trading amounts $\psi^\circ_1(t)$. All plots are based on 1000 randomly drawn trajectories of $\{R(t)\}, t=1,\ldots, 25$. \label{fig:Banking-price}}
\end{figure}

To understand the joint behavior of the two farmers, Figure \ref{fig:joint} shows the distribution of $C^\circ_2(t)$ vs $C^\circ_1(t)$ (as well as of $b^\circ_2(t)$ versus $b^\circ_1(t)$) for three representative periods $t$. For the banking amounts, there is a strong correlation for earlier periods, cf.~$t=6$, and less correlation as $t$ increases, due to the path-dependent features that emerge intrinsically in our setup, in particular due to the banking caps. In contrast, there is a fixed dependence between equilibrium pumping $C^\circ_2(t)$ and $C^\circ_1(t)$ based on the (time-stationary in our setup) split of available water between the two agents. Indeed, given total consumption $\cW$, agents trade among themselves to align their $C^\circ_1(t), C^\circ_2(t)$ based on their relative productivity and profitability. 

\begin{figure}[!htb]
\centering
\includegraphics[width=0.475\textwidth]{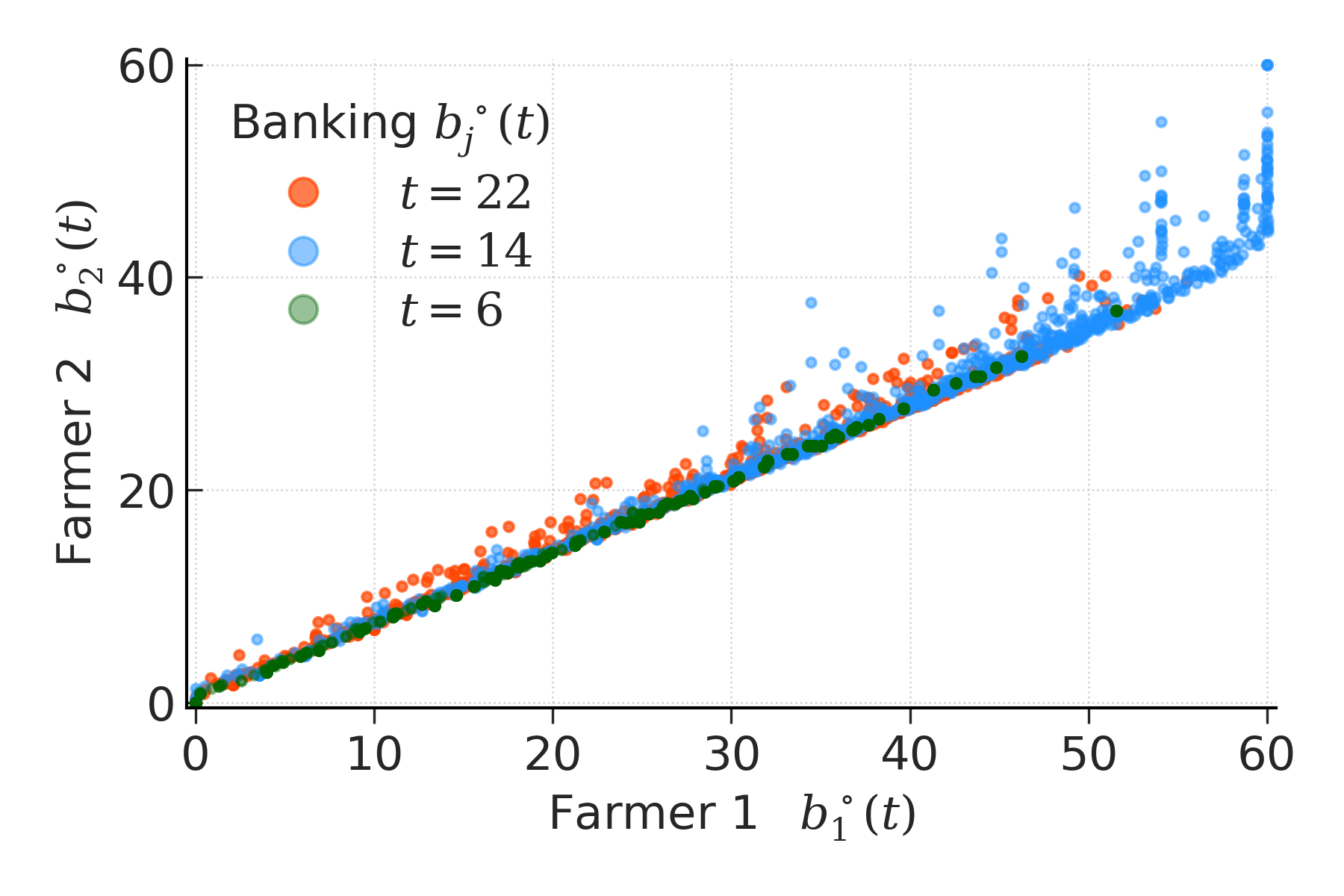}
\includegraphics[width=0.475\textwidth]{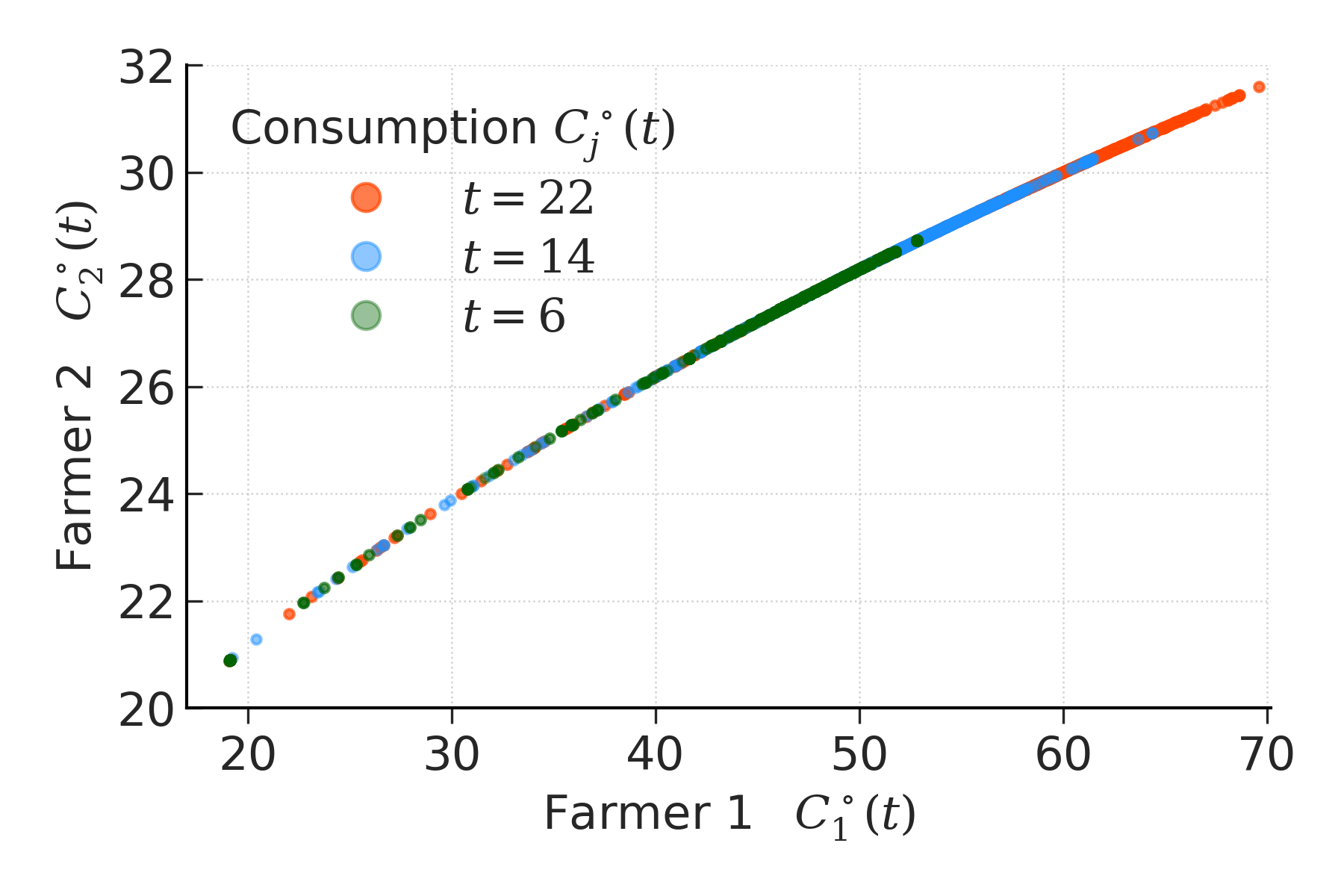}
\caption{Left: empirical distributions of equilibrium $b^\circ_2(t)$ vs $b^\circ_1(t)$; Right: empirical distribution of equilibrium pumping amounts $C^\circ_2(t)$ vs $C^\circ_1(t)$, at $t=6, 14, 22$. Both plots are based on 1000 randomly drawn trajectories of $\{R(t)\}$. \label{fig:joint}}
\end{figure}

\subsection{Gains from Trade and from Banking}

Groundwater banking allows agents to smooth out the stochastic fluctuations in water allocations, while trading allows them to leverage comparative advantages in productivity. We can tease out the relative value of these two effects by constraining the respective markets. In the first case, we can allow agents to individually bank their water rights across time but do not give them the option to trade. In the second case, we allow trading but no banking -- equivalent to a sequence of one-period problems, already studied in \cite{CialencoLudkovski2025}. 

Figure~\ref{fig:profit-distrib} compares the empirical distribution of one-step profits $v_j(\bfW^\circ(t))$ relative to the base situation where banking is not allowed. In the latter case, farmers must consume whatever is allocated, leading to net profit of $v_j(\mathbf{R}(t))$; since $R(t)$ takes just three values this leads to a discrete distribution with a 3-point support, namely $v_j(\Theta_j(r_m)), m=1,\ldots,M$ with respective probabilities equal to the stationary distribution $\vec{\pi}$ of $\{ R(t)\}$, see the red bars in the Figure~\ref{fig:profit-distrib}.  With banking, as already discussed in connection with Figure~\ref{fig:20-period}, farmers conserve and bank groundwater rights during wet years and spend down those carry-forwards during droughts. Their resulting profits now have a non-trivial distribution, concentrated in $v_1 \in [60,74]$ and $v_2 \in[52,59]$, which is generally more than during the medium-precipitation year and less than during the wet year. Three effects can be observed. First,  profit variability is dramatically reduced. Compared to the case of no banking, the standard deviation is reduced by 40\% for Farmer 1 and by 36\% for Farmer~2. Second, this reduction is primarily in the left-tail: the base drought-driven profits are almost never observed once banking is allowed: this happens just in 6 out of 1000 simulations. Indeed, the 95\% profit quantile is 60.05 for Farmer 1 and 50.02 for Farmer 2, whereas without banking there was $\pi_1 = 1/9 = 11.1\%$ of earning less than 40. Nevertheless, the long left tail, caused by the possibility if multiple consecutive dry years, remains. Third, the precautionary motive is driven by risk-aversion: in the base case the strongly concave log-utility creates a tighter and less-skewed profit distribution compared to risk-neutral farmers.

\begin{figure}[!htb]
\centering
\includegraphics[width=0.9\textwidth]{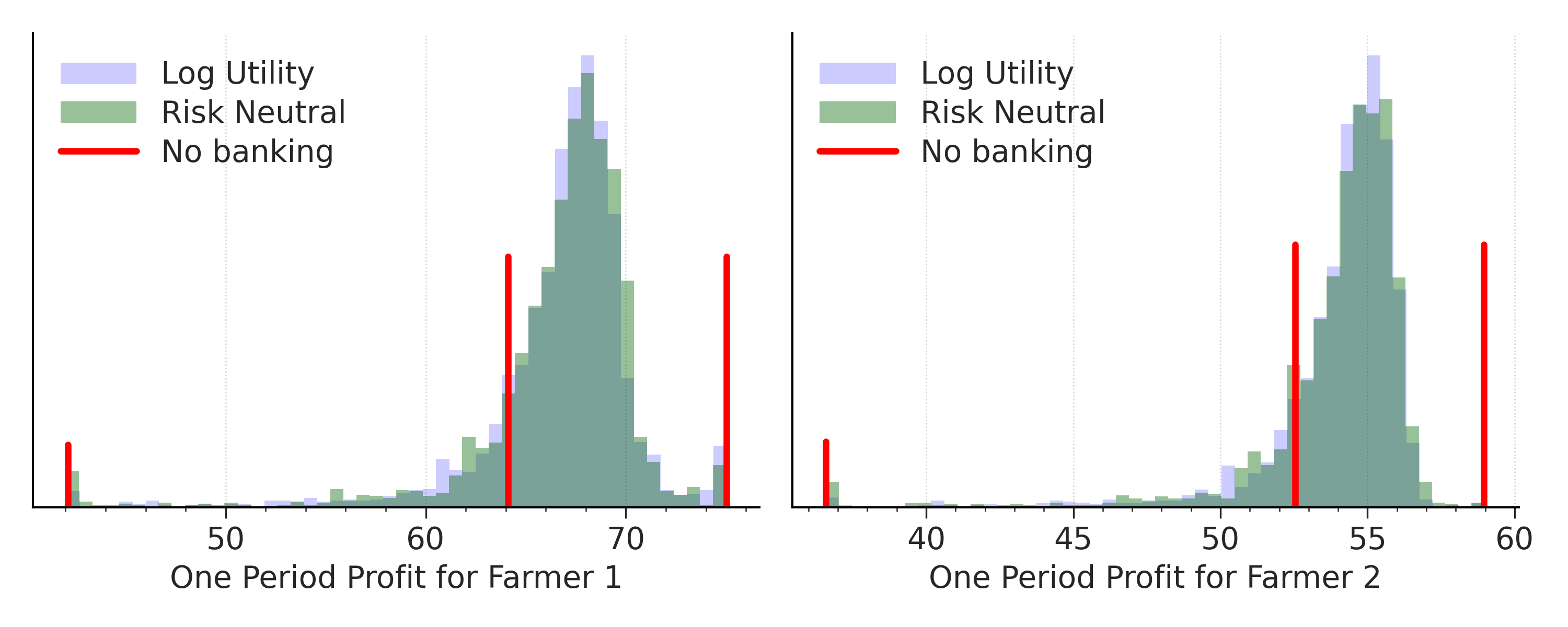}
\caption{Distribution of annual profits of Farmer 1  and 2, $v_j(t,\bfW^\circ(t))$, based on 1000 trajectories of $\{R(t)\}$ and averaged across $t \in\{10,11,12\}$. The red bars represent the distribution of profits in a market where no banking is possible. \label{fig:profit-distrib}}
\end{figure}

Another useful comparison is relative to the model with a single agent. First, this captures the possibility of individual farmers who can bank but cannot trade with each other, so that multiple single agents operate in parallel, without any interaction among them. Second, a monopolist is relevant in the context of a central planner. The central planner optimizes aggregate welfare, bypassing all and any market frictions and the cost of the non-cooperative strategies. 

Specifically, the single agent problem is the one-dimensional version of \eqref{eq:farmer1-2period}. Each period, farmer $j$ optimizes their period-$t$ banking amount $\beta_j$ (with the hard period-$t$ water budget constraint of $C_j(t) = w- \beta_j$) to maximize total expected profits 
\begin{align} \label{eq:1agent}
    J_j(t, w, r) := \sup_{\beta_j \in [0, w-\underline{C}]} \Big\{ U_j(G_j(t,w- \beta_j)) + e^{-\delta}\bE \left[ J_j(t+1, R_j(t+1) + \beta_j, R(t+1)) \, |\ R(t)=r \right] \Big\}.
\end{align}
To solve \eqref{eq:1agent} we initialize with 
$$\widehat{J}_j(T,w,r) = U_j(G_j(t,w))$$
and recursively (as $t=T-1, T-2,\ldots$) construct one-dimensional surrogates (such as a cubic spline) $w \mapsto \widehat{J}_j(t,w,r)$ which are trained by numerically solving \eqref{eq:1agent} on a grid of training allocations and in parallel for different $r$'s. 

In the absence of the market, Farmer 2 no longer has the option to buy rights from Farmer 1 during droughts; as a result the variance of her consumption goes up significantly (and the variance of $C^\circ_1(t)$ goes down). As a further consequence, Farmer 2 needs bigger reserves to control her downside risk of losses. We therefore find that $\bE[ \beta_2(t) ] > \bE[ b^\circ_2(t)]$ (on the order of several percent), while there is minimal change to the banking strategy of Farmer 1, see last row of Table \ref{tab:ave-banking}. This reveals asymmetric benefit of market access for different stakeholders.

\section{Comparative Statics}\label{sec:statics}

To illustrate how the agents adjust their behavior relative to the market and rainfall parameters,  in this section we provide several comparative analyses that tease out the impact of different model parameters. As our base setup we use the above case study. Beyond the figures below, Table~\ref{tab:ave-banking} summarizes the mean and standard deviation of banking amounts, one-period (risk-unadjusted) profits and groundwater prices across the range of experiments considered. To capture the long-run steady-state market conditions, we average across $t \in\{10,11,12\}$ always using a horizon of $T=25$ years. 
\begin{table} \small
$$\begin{array}{rrrrrr} \hline 
\text{Case} & \bE[p^\circ(t)] &  \bE[b_1^\circ(t)] &  \bE[b^\circ_2(t)] & \bE[v_1(\bfw^\circ(t))] & \bE[v_2(\bfw^\circ(t))] \\ \hline  
\text{Base} &  0.90 \ (0.052) & 36.39 \ (15.127) & 26.65 \ (11.621) & 66.84 \ (4.078) & 54.13 \ (2.442) \\
\bfQ=\bfQ_2 & 0.90 \ (0.064) & 33.26 \ (17.475) & 25.71 \ (14.125) & 66.66 \ (4.888) & 53.93 \ (2.890) \\
\bfQ=\bfQ_3 & 0.90 \ (0.037) & 32.32 \ (14.761) & 22.86 \ (10.740) & 66.91 \ (3.111) & 54.20 \ (1.814) \\ \hline
U_1(x)= U_2(x) = \sqrt{x} & 0.90 \ (0.056) & 33.32 \ (15.528) & 24.50 \ (11.627) & 66.82 \ (4.182) & 54.11 \ (2.586) \\
U_1(x)=U_2(x)=x & 0.90 \ (0.060) & 31.68 \ (15.781) & 23.96 \ (11.994) & 66.76 \ (4.408) & 54.04 \ (2.729) \\
U_1=\tilde{u}, U_2=\textrm{Id} & 0.90 \ (0.063) & 36.24 \ (15.593) & 23.07 \ (12.232) & 66.70 \ (4.533) & 54.09 \ (3.016) \\
\rq^{(2)} =[1/3, 1/3, 1/3] & 0.97 \ (0.063) & 33.77 \ (16.761) & 27.93 \ (14.567) & 61.16 \ (4.925) & 50.57 \ (2.641) \\
\rq^{(3)} = [4/9, 4/9, 1/9] & 1.04 \ (0.064) & 28.78 \ (16.540) & 22.84 \ (13.876) & 56.01 \ (4.433) & 47.16 \ (2.593) \\
R(t) \in \{36, 71, 100\} 
& 0.90 \ (0.063) & 38.70 \ (15.907) & 31.79 \ (14.099) & 67.06 \ (5.151) & 54.14 \ (2.771) \\
\hline
\Theta = [0.4, 0.6] & 0.90 \ (0.054) & 26.59 \ (11.834) & 35.96 \ (15.484) & 52.38 \ (3.289) & 68.32 \ (3.368) \\
\text{Senior/Junior} & 0.90 \ (0.065) & 17.74 \ (7.842) & 39.47 \ (15.915) & 62.66 \ (2.570) & 57.89 \ (5.884) \\
\hline 
\delta=0.02 & 0.90 \ (0.069) & 22.55 \ (11.688) & 16.32 \ (8.045) & 66.90 \ (5.040) & 54.16 \ (3.202) \\
\delta=0.04 & 0.90 \ (0.084) & 13.89 \ (8.481) & 10.49 \ (5.952) & 66.77 \ (6.021) & 54.05 \ (3.843) \\
h=0.05 & 0.90 \ (0.090) & 13.01 \ (8.340) & 10.06 \ (5.946) & 66.67 \ (6.418) & 53.97 \ (4.065) \\ \hline
\bar{b}_1=0 & 0.90 \ (0.085) & 0.00 \ (0.000) & 30.62 \ (13.412)   & 66.41 \ (9.897) &  54.17 \ (2.290) \\
\bar{b}_2=0 & 0.90 \ (0.078)  & 36.60 \ (15.301) &  0.00 \ (0.000) & 66.90 \ (4.025) &  53.75 \ (6.581)   \\
\bar{b}_2=10 & 0.90 \ (0.073) & 36.25 \ (16.140) & 7.21 \ (3.129) & 66.71 \ (4.454) & 54.01 \ (4.865) \\
\text{No trading} & - & 35.85\ (15.624) & 27.81\ (12.142 ) & 66.57\ (4.468) & 53.66\ (2.596)  \\
\hline
\end{array}$$
\caption{Average groundwater price $p^\circ$ and one-period profits across considered market setups. The base case uses $\bfQ = \bfQ_1, R(t) \in \{40, 75, 95\}$, $U_1(x)=U_2(x) = \tilde{u}(x) =\log(x-20), \Theta = [0.6, 0.4], \delta=0, h=0, \bar{b}_1 = \bar{b}_2 = 60$. Values in parentheses indicate the respective standard deviations across 1000 scenarios. \label{tab:ave-banking}}
\end{table}

\subsection{Environmental Uncertainty}

Stochastic groundwater allocations, driven by the random groundwater recharge, is a central piece of our modeling framework. To study the resulting impacts, we adjust different components that affect the distribution of $R(t)$. 

\textbf{Recharge chain:} we first consider modifying the transition matrix $\bfQ$ that governs the serial correlation in $\{R(t)\}$. To do so, we fix the stationary distribution $\vec{\pi}_R=[1/9, 4/9, 4/9]$ which represents the long-run groundwater availability and vary the persistence of the current regime. We consider the following two alternatives to $\mathbf{Q}_1$:
\begin{align}\label{eq:Q-cases}
\mathbf{Q}_2 = \begin{bmatrix} 1/2 & 1/2 & 0 \\
                        1/72 & 47/72 & 1/3 \\
                        1/9 & 2/9 & 2/3 \end{bmatrix}, \qquad 
\bfQ_3 = \begin{bmatrix} 1/9 & 4/9 & 4/9 \\ 
                        1/9 & 4/9 & 4/9 \\
                        1/9 & 4/9 & 4/9 \end{bmatrix}.
\end{align}
In the latter case with $\bfQ_3$, there is no memory and precipitation regimes are i.i.d.~across years. Therefore the current year's regime does not impact the distribution of next year's recharge and $V^\circ_j(t,\bfw,r)$ is independent of $r$. This allows to quantify the impact of auto-correlation in $\{R(t)\}$. In the second case, $\mathbf{Q}_2$ is chosen such that $\{R(t)\}$ exhibits a strong persistence: in particular a dry year $R(t) =40$  is likely to be followed by another dry year, while a medium-rainfall year is likely to be followed by another medium year. Note that in the base case with $\bfQ_1$, next-year recharge distribution is identical in states 2 and 3 (hence will lead to identical banking decisions), but there is more persistence in the dry state, making consecutive droughts more likely. Hence the baseline $\bfQ_1$ presents an intermediate situation with moderate  serial correlation.

Lack of persistence discourages banking, as there is less risk of multi-year droughts, hence under $\mathbf{Q}_3$ equilibrium banking quantities are shifted down (especially noticeable for Farmer 2). In parallel, the variability of profits $\textrm{StDev}(v_j(\bfW^\circ(t))$ drops significantly (by 28\%) under $\mathbf{Q}_3$ compared to the base case. Conversely, the more persistent system with $\mathbf{Q}_2$ creates more ``feast-famine'' situations where carryover is either exhausted or maxed out, raising the standard deviation of $v_j(\bfW^\circ(t))$ by 17\%. In Figure \ref{fig:recharge-distrib}a we similarly observe that the variance of $b_j^\circ(t)$ increases under $\mathbf{Q}_2$ and agents become more sensitive to the current regime (cf.~more separation of colors that represent $R(t)$ state: red=dry, yellow=medium, gray=wet).  Because a drought is very unlikely after a medium year under $\mathbf{Q}_2$, corresponding banking in average rainfall years drops, while there is relatively more banking in dry years.

\begin{figure}[!htb]
\centering
\includegraphics[width=0.49\textwidth]{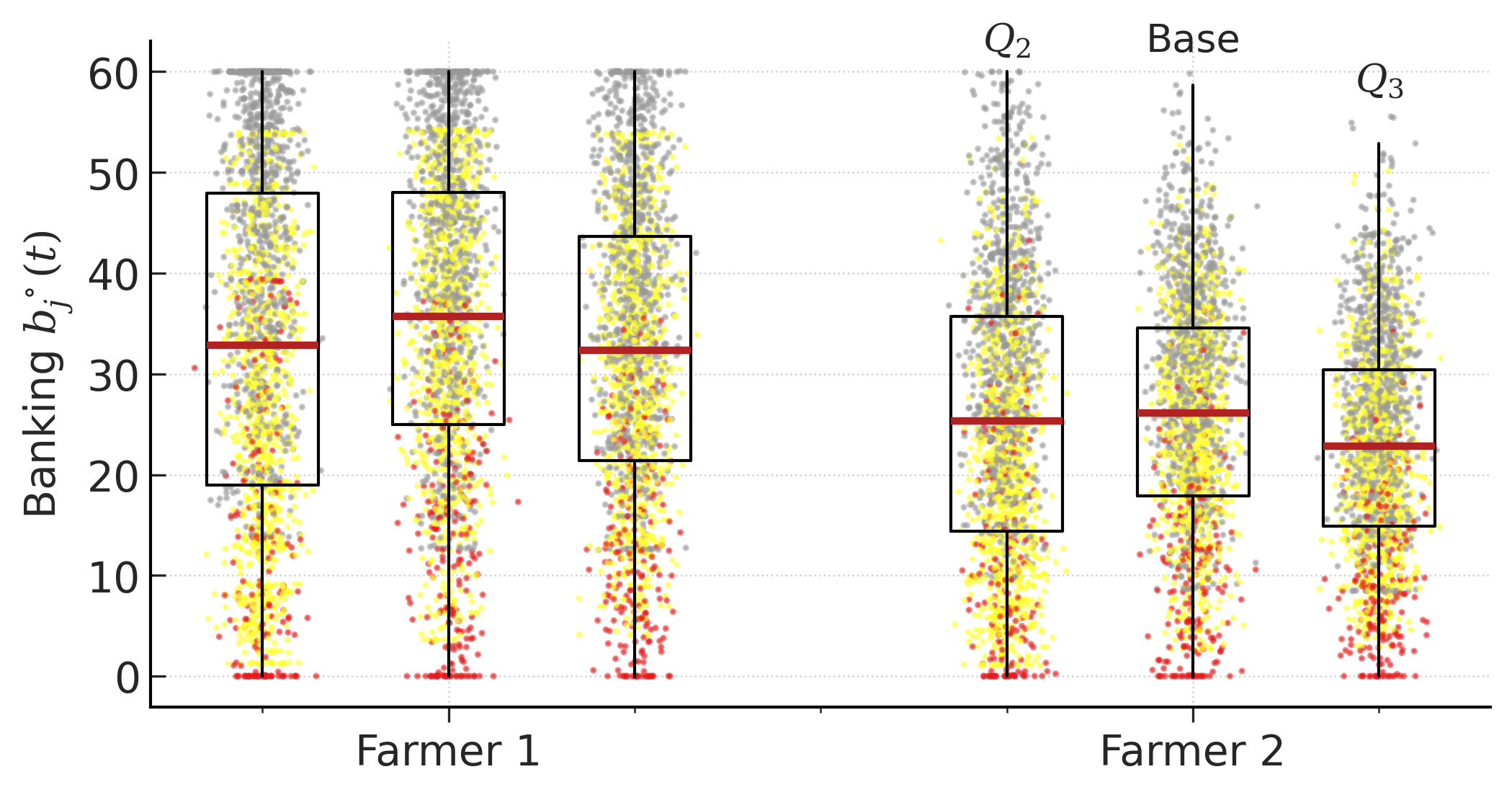}
\includegraphics[width=0.49\textwidth]{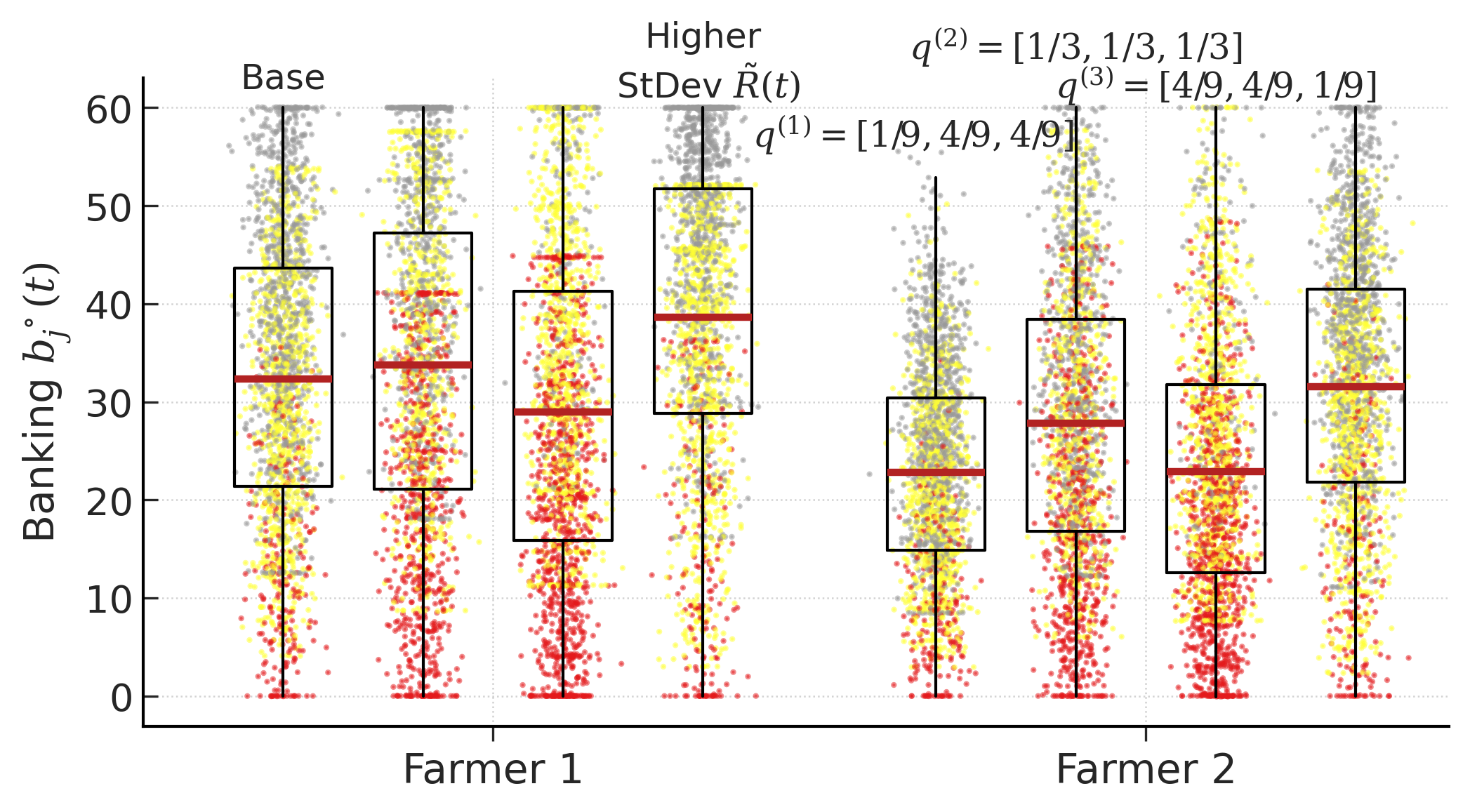}
\caption{Distribution of equilibrium banking $b^\circ_j(t)$, color-coded by the regime $R(t)$. The boxplots approximate the steady-state market conditions (average across $t \in\{10,11,12\}$); colors indicate the current recharge regime (red=dry, yellow=medium, grey=wet).  \emph{Left panel (a):}  for different $\bfQ$-matrices from \eqref{eq:Q-cases}. \emph{Right (b):} for different recharge distributions $\rq$ postulating i.i.d.~$\{R(t)\}$. \label{fig:recharge-distrib}}
\end{figure}

\textbf{Recharge Distribution:} To isolate the impact of climate, we consider the setting of i.i.d $\{R(t)\}$ and then modify the  relative frequency of the 3 weather regimes. The base case is $\rq^{(1)} = \vec{\pi}_R = [1/9, 4/9, 4/9]$ with average recharge of 80 ac-ft. In the right panel of Figure \ref{fig:recharge-distrib} we also consider $\rq^{(2)} = [1/3, 1/3, 1/3]$ (drier climate with average recharge of $\bE[R(t)]=70$ ac-ft) and $\rq^{(3)} = [4/9, 4/9, 1/9]$ (very dry climate with average recharge of 61.7 ac-ft). As droughts become more frequent and allocations are reduced, equilibrium prices rise and their variability $\textrm{StDev}(p^\circ(t))$ increases. We observe average price increase of 16\% and standard deviation increasing by 89\% comparing $\rq^{(3)}$ to $\rq^{(1)}$, cf. Table \ref{tab:ave-banking}. Under very dry climate, farmers strongly curtail consumption during dry years, slowly drawing down reserves to maintain some buffer against the possibility of multi-year shortages, cf.~many scenarios where $b^\circ_j(t)>30$  in the dry regime under $\rq^{(3)}$ in Figure \ref{fig:recharge-distrib}b.

A different effect is obtained if we modify the values $r_m$ that $R(t)$ can take on, keeping the respective probabilities $\rq$ fixed. This can be interpreted as a distributional shift in precipitation and groundwater recharge, rather than a change in the frequency of droughts. Many climate projections for Western US suggest that dry years will get dryer and wet years will get wetter, increasing the range of annual precipitation. To capture the impact of such a change, we consider 
\begin{align}\label{eq:R-dist-2}
\tilde{R}(t) = \left\{ \begin{aligned}
    36 & & \text{with prob. } \rq_1= 1/9;\\
     71 & & \text{with prob. }\rq_2=4/9;\\
      100 & & \text{with prob. } \rq_3=4/9,
\end{aligned}  \right. \end{align}
which has the same $\bE[\tilde{R}(t)] = 80$ as our base case, but higher variance, in particular droughts are now 10\% more severe. The respective banking distributions are in the right-most column of Figure \ref{fig:recharge-distrib}b and demonstrate that more variable recharge raises the intertemporal smoothing motive and farmers tend to carry higher reserves.

\subsection{Groundwater Rights Allocations}

\begin{figure}[!htb]
\includegraphics[width=0.49\textwidth]{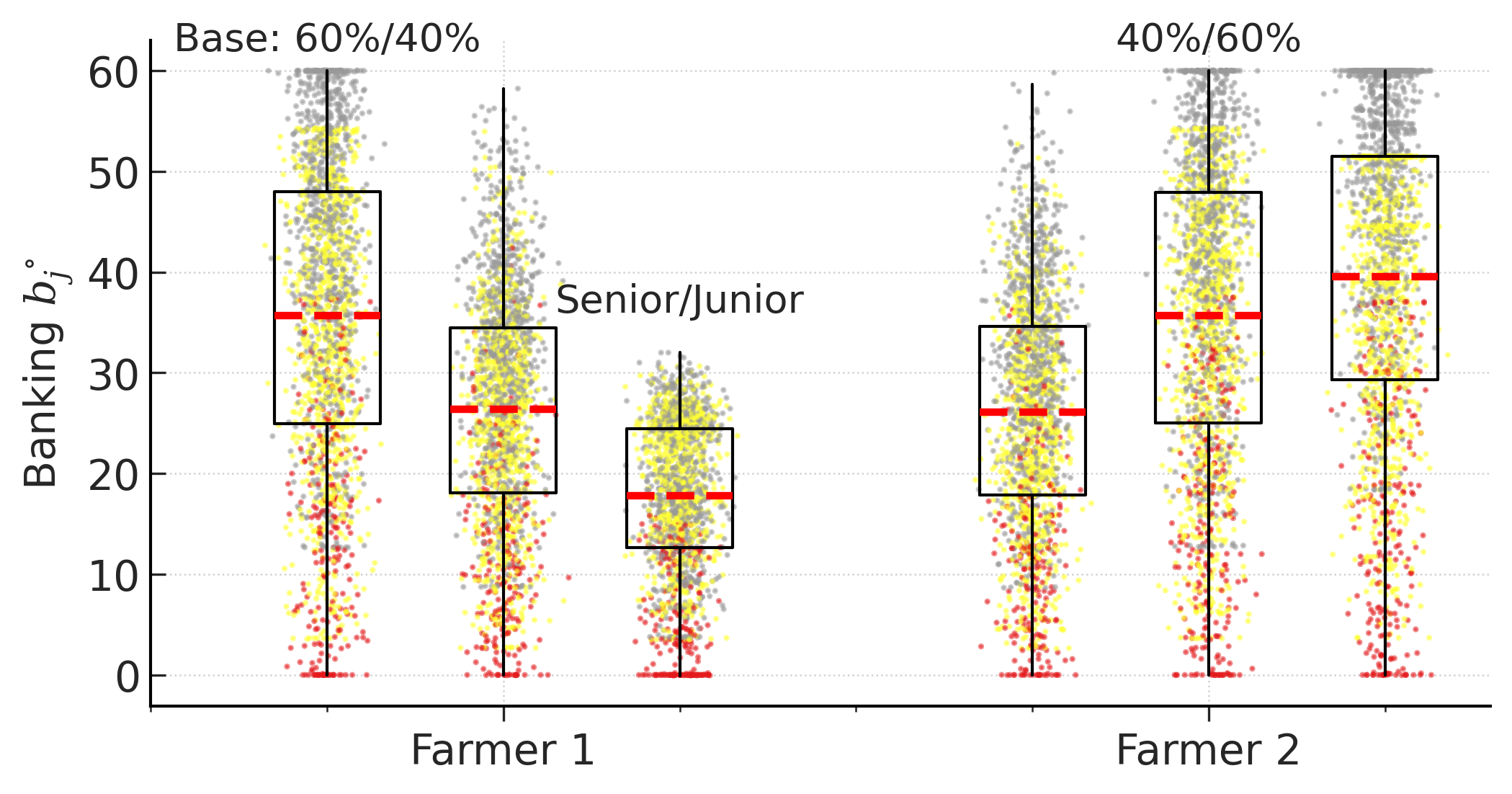}
\includegraphics[width=0.49\textwidth]{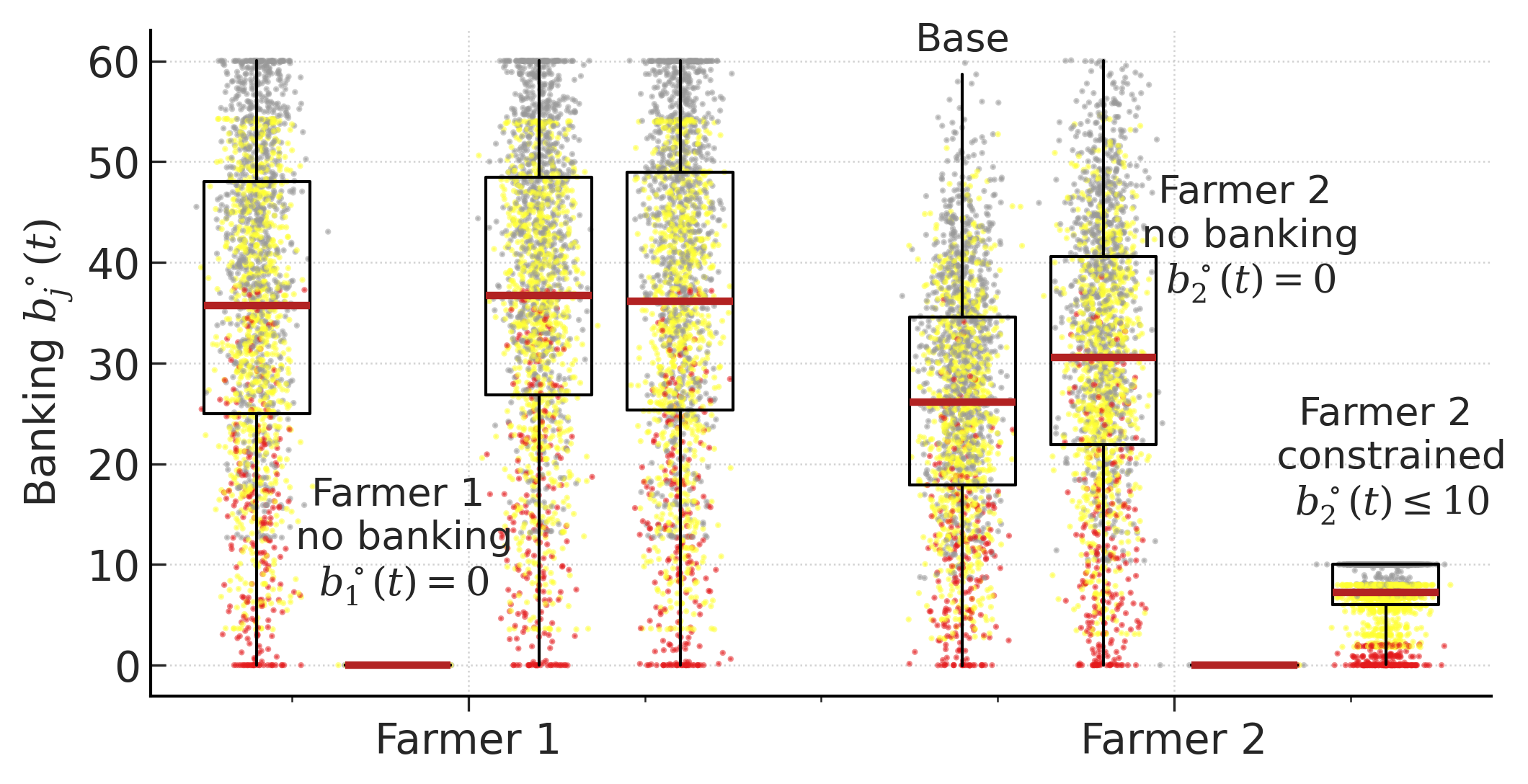}
\caption{Impact of allocation mechanisms $\Theta_j(r)$ on distribution of equilibrium banking $b^\circ_j(t)$. \emph{Left panel (a):} Three different allocation mechanisms. \emph{Right (b):} flexibility constraints. All plots based on 1000 simulated trajectories of $\{R(t)\}$, color-coded by the current regime, cf.~Figure \ref{fig:recharge-distrib}. \label{fig:alloc-impact}}
\end{figure}

Next, we analyze the impact of allocations. Fixing the distribution of $R(t)$ and the base case of (a) ``60/40'' allocation: $\Theta_1(t) = \theta_1 R(t) \in \{24, 45, 57\}, \Theta_2(t) = \theta_2 R(t) \in \{16, 30, 38\}$, we consider the alternatives of (b) a ``40/60'' allocation where Farmer 2 gets the majority of water: $\theta_1 R(t) \in  \{16, 30, 38\}, \theta_2 R(t) = \{24, 45, 57\}$
and a (c) ``Senior/Junior'' allocation scheme  where Farmer 1 is guaranteed at least 30 ac-ft, and hence gets nearly all the water during droughts but does not benefit from a wet year, $\Theta_1(t) = \{30, 45, 45\}$ while the junior Farmer~2 has a highly variable allocation $\Theta_2(t) \in \{10, 30, 50\}$, absorbing nearly all the fluctuations in consumption. The left panel of Figure~\ref{fig:alloc-impact} shows that changing who gets the water drastically impacts who is banking their water rights. Intuitively, with total available water fixed, the best use of water is invariant to the allocation. If Farmer~1 is getting more water then she will bank more and Farmer~2 will bank less. In the ``40/60'' allocation, Farmer 2 becomes the main ``banker''. Similarly, if the variability of $\Theta_1(t)$ is reduced (case c), then Farmer 1 has less incentives to bank, while the ``Junior'' Farmer 2 will bank a lot in order to smooth out her extreme fluctuations. The shift in the allocation also induces a shift in profits -- when Farmer 2 gets a larger fraction of supply, their average profit goes up and that of Farmer 1 goes down, cf.~Table \ref{tab:ave-banking}.

\textbf{Consumption Flexibility:} In the base case, both farmers have the same flexibility and only differ in their profit functions. To get insights into the value of flexibility we consider the situation where only one out of the two agents is allowed to bank. In that case, there is no banking equilibrium to determine: for instance if only Farmer 1 can bank, then she can determine her optimal $b^\circ_1(t,\bfw,r)$ as a direct optimization, similar to \eqref{eq:1agent}. Conditional on $b^\circ_1$, the resulting current-period profits of both Farmers, and trading between them, are determined from the 1-period equilibrium. The right panel of Figure \ref{fig:alloc-impact} shows the impact of banking restrictions is asymmetric. If Farmer 1 cannot bank, then Farmer 2 will bank even more to compensate. However, if Farmer 2 cannot bank, then Farmer 1 does not really modify her behavior. We moreover see that if Farmer 2 has a low upper cap on carry-forward banked rights, then she will often draw down her reserves, in particular use up her entire bank during a dry year.

\subsection{Agent Preferences}

A third aspect that drives stakeholder actions concerns their risk and intertemporal preferences. Since banking is inherently linked to trading off present and future (uncertain) consumption, risk-aversion and intertemporal discounting strongly impact the choice of how much pumping to forgo today and carry forward into the next period. 

\begin{figure}[!htb]
\centering
\includegraphics[width=0.49\textwidth]{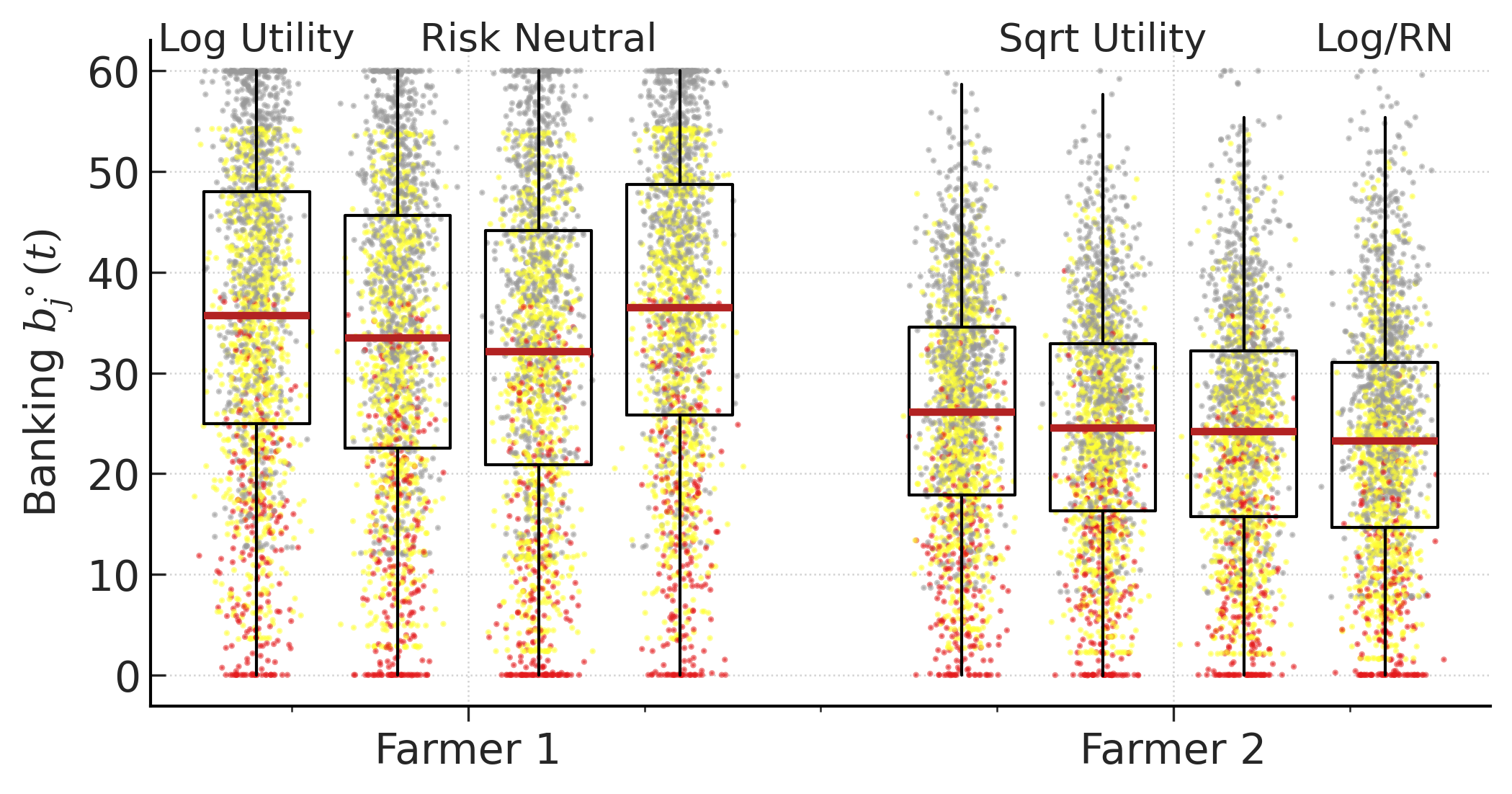}
\includegraphics[width=0.49\textwidth]{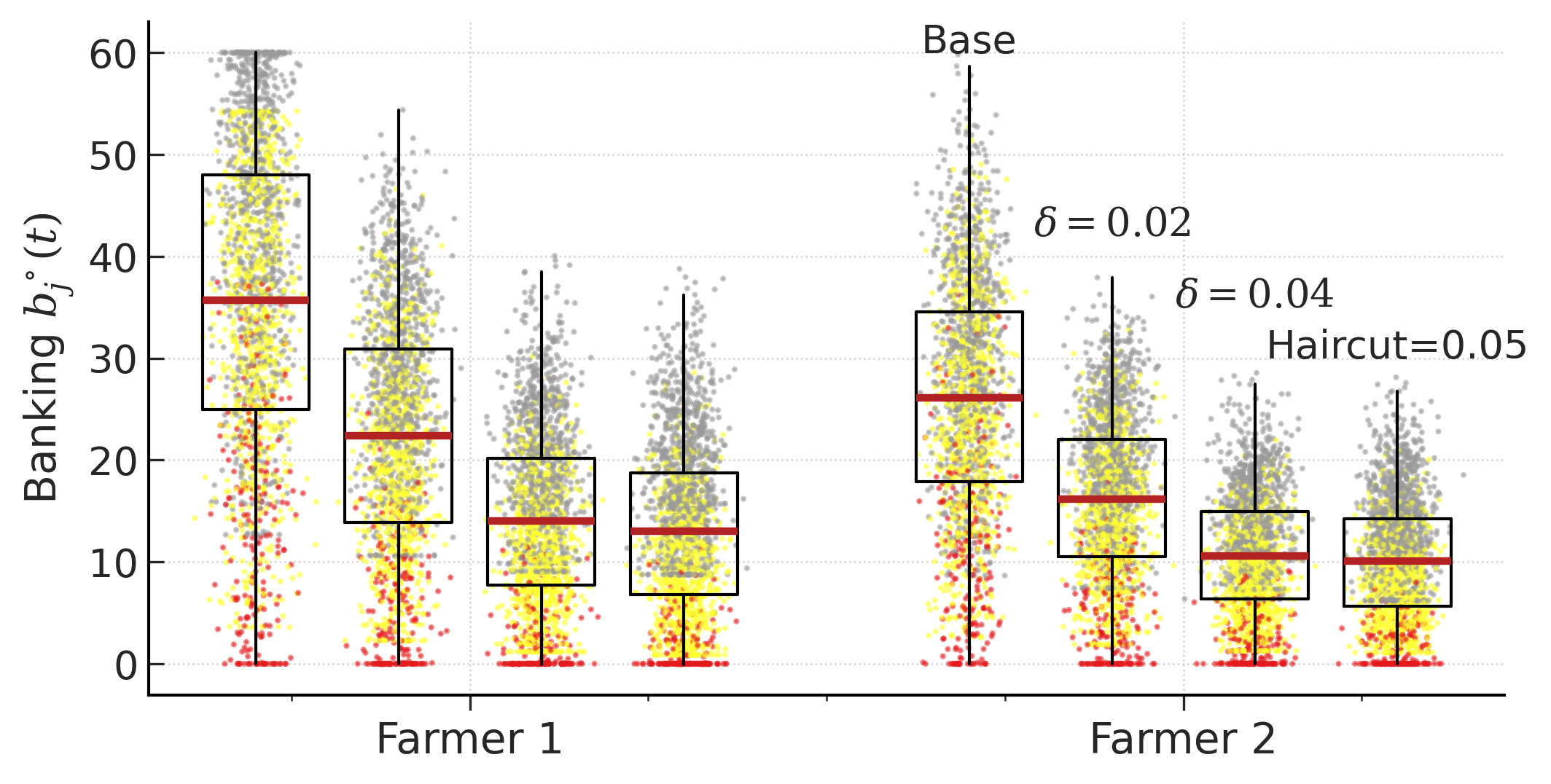}
\caption{Impact of agent preferences on equilibrium banking $b^\circ_j(t)$. \emph{Left panel (a):} different utility functions.  \emph{Right panel (b):} different discount factors $\delta$ and banking tax $h=0.05$ in \eqref{eq:haircut}. The Base case uses $\delta=0$ and $U_1(x) = U_2(x) = \log(x-20)$. All plots based on 1000 simulated trajectories of $\{R(t)\}$. \label{fig:utility-impact}}
\end{figure}

\textbf{Risk aversion:} The left panel of Figure \ref{fig:utility-impact} visualizes the role of the utility functions $U_j(\ell)$. We compare our base case with log-utility to two other settings where agents are more risk tolerant, namely (i) both agents having $U_j(x) = \sqrt{x}$ square-root power utility and (ii) both agents being risk-neutral $U_j(x) = x$. Risk-aversion is one of the key causal mechanisms for banking, amplifying the losses due to lower consumption in droughts, which motivates precautionary saving of rights. Thus,  as expected, lower risk-aversion progressively reduces banking relative to the base case as we assume square-root and linear utility. As a final comparison, we take $U_1(x) = \log(x -20)$ and $U_2(x)=x$ to capture risk aversion heterogeneity, see rightmost column in Figure \ref{fig:utility-impact}a. In the latter situation we observe a relative shift in banking with the more risk-averse Farmer 1 banking 50\% more than the risk-neutral Farmer 2.

\textbf{Intertemporal discount rate:} Another factor that controls the benefit of carry forward balances is discounting: agents that are fully myopic and solely care about the present have no reason to bank. Thus, higher discount rates $\delta$ will induce progressively less banking. This is illustrated in the right panel of Figure \ref{fig:utility-impact} where we compare the base case with no discounting $\delta=0$ against $\delta=0.02$ and $\delta=0.04$. In the latter case, average amount of banking is almost halved. In particular, discounting strongly discourages building up large reserves which would take several years to wind down. While in the base case Farmer 1 carries over 40 ac-ft of reserves more than 40\% of the time, with $\delta=0.02$ she does so only  with probability of 7\% and with $\delta=0.04$, she never does so, never accumulating more than 39 ac-ft of reserves.

\textbf{Tax on banking:} due to lateral groundwater flows and interbasin connections, regulators often impose a ``haircut'' that applies a proportional tax of $h$\% on carryover amounts:
\begin{multline}\label{eq:haircut}
      V^\circ_j(t,\bfw, r):=   v_j \left(\bfw - \bfb^\circ(t,\bfw,r) 
    \right) \\ + e^{-\delta} \cdot \bE \left[ V^\circ_j \left(t+1, 
    \bfR(t+1) + (1-h)\bfb^\circ_1(t,\bfw,r), R(t+1) \right) \Big| \ R(t) = r\right].
\end{multline}
We illustrate the effect of imposing a 5\% haircut on banked water in the rightmost column of Figure~\ref{fig:utility-impact}b. As expected, such tax dramatically lowers the benefits of carryover, with impact similar to a discount factor.

\section{Calibrated Case Study}\label{sec:calilbration}

Finally, we present a setup roughly calibrated to real-life conditions, namely the San Joaquin river region of California's Central Valley. This is one of the principal agricultural regions in the U.S., where farming is one of the major regional economic engines and groundwater pumping is extensively utilized to supplement precipitation and surface water irrigation.

California climate features well-defined rainy and dry seasons, with nearly all measurable precipitation taking place during October-April. This seasonality simplifies the definition of the precipitation regime process which is based on the aggregate rain amounts during the winter. Using historical hydrological data, a 3-state Markov chain has been calibrated in \cite{Inston2025} to capture the annual precipitation pattern in the San Joaquin river hydrologic region. Next, using the C2VSimFG \cite{C2V} tool from the California Department of Water Resources (which is also used in the aforementioned CALVIN optimization framework), we ran a linear regression to estimate how precipitation maps to groundwater recharge $R_{SJ}(t)$. This leads to the state space with respective transition matrix
\begin{align}\label{eq:SJ}
R_{SJ}(t) \in \{ 2000, 2800, 3500\}, \qquad
\bfQ_{SJ} = \begin{bmatrix} 
0.20 & 0.28 & 0.52 \\
0.21 & 0.36 & 0.43 \\
0.11 & 0.13 & 0.76 \\ 
\end{bmatrix}.
\end{align} 
Above $R_{SJ}(t)$ is in 1000s of ac-ft and represents recharge regime across the entire basin. Since actual recharge has a continuous distribution, following \cite{Inston2025} we take $R_j(t)|R_{SJ}(t)=r$ to be log-normal, with the matching mean and standard deviation of 8\%. For numerical purposes, the respective expectation in  \eqref{eq:quantized-expectation} is implemented as an additional inner sum over $\ell$ (we use 5 terms)
\begin{multline*}
\bE \left[ \hat{V}_j \left(t+1, 
    \bfR(t+1) + \bfb^\circ(t,\bfw,r), R(t+1) \right) \Big| \ R(t) = r_k\right] =  \\ \sum_{m=1}^M \bfQ_{k,m} \sum_\ell e_\ell \widehat{V}_j \left(t+1, \Theta(r_{m,\ell})+ \bfb^\circ(t,\bfw,r), r_m\right),
\end{multline*}
where $e_\ell$ are the Gaussian quantization weights and $r_{m,\ell}$ are the quantized values of the log-normal distribution with mean $r_m$.

Currently none of GSAs in San Joaquin basin implement marketplaces for trading pumping rights, however several agencies impose quotas and other allocation mechanism to manage overdrafts.  To outline the potential groundwater market, we consider two types of farming stakeholders: tree orchards and forage. Specifically, we consider olives and almonds which are widely cultivated in the San Joaquin basin and require significant and regular irrigation, as well as alfalfa that is seasonally planted. 
We use yield values reported by \cite{SearsEtAl2019} who calibrated a power-type profit function of the form \eqref{eq:G}. Specifically, they interpret all quantities on a per-acre basis, with variable irrigation quantity $\phi^k_j(t)$ representing the amount of water applied to an existing productive patch that grows crop $k$. Irrigation is interpreted relative to a baseline level $\bar{\phi}^k$, with $\phi^k_j(t) < \bar{\phi}^k$ corresponding to ``deficit irrigation'', whereby the patch is irrigated less than needed, producing a smaller harvest. The multiplier $\bar{f}^k$ in \eqref{eq:G} is the dollar value of harvested yield per acre and $q^k$ is the proportional cost of pumping and irrigating crop $k$. The values below are in \$100's per acre with the denominators corresponding to the baseline $\bar{\phi}^k$'s of each crop with $\phi_j$ measured in ac-ft:
\begin{itemize}
    \item Alfalfa: $F_j(t,\phi_j) = 25 (\frac{ \phi_j}{5.3})^{0.3} - 1.3 \phi_j$;

    \item Almonds : $F_j(t,\phi_j) = 63 (\frac{\phi_j}{3.67})^{0.091} - 1.0 \phi_j$;

    \item Olives: $F_j(t, \phi_j) = 51 (\frac{\phi_j}{4.0})^{0.143} - 0.9 \phi_j$.
 
\end{itemize}
Note the relatively low elasticities of production $\alpha^k$ and the much higher profits from orchards compared to alfalfa. Above we omit any fixed costs, so $F_j$ should not be interpreted as net profit, only as net variable revenue. 

We assume that Farm 1 only grows alfalfa, while Farm 2 has equal acreage of olives and almonds and is double the size of Farm 1. Moreover, we assume fixed per-acre allocation, so that $R_j(t)$ is proportional to $R(t)$. To convert from total basin recharge above to per-acre quantities, we multiply by $1.2 \cdot 10^{-3}$. The San Joaquin basin has about 2.1M irrigated acres, however most of the irrigation is conjunctive with surface water supply. Given that this setup is fully proportional to acreage, we assume without loss of generality that Farm 1 consists of $A$ acres and Farm 2 of $2A$ acres, so that $R_2(t) \equiv 2 R_1(t)$.  Finally, we assume risk-neutral Farmers that use a discount rate of $\delta=3\%$ per year and that banking is capped at $5$ ac-ft for both stakeholders.

The values in \eqref{eq:SJ} imply that about 66.1\% of the years are wet, 19.7\% are medium and 14.2\% are dry, with average per-acre allocation of 3.78 ft, which is about 87\% of the nominal average baseline needs for the three crops considered. In particular, during wet years the allocation is 97\% of nominal, but it is only 55\% during droughts. Thus, farmers must engage in some deficit irrigation in order to build up reserves and mitigate droughts.

\begin{figure}[ht]
\includegraphics[width=0.73\textwidth,trim=0in 0.3in 0in 0.7in]{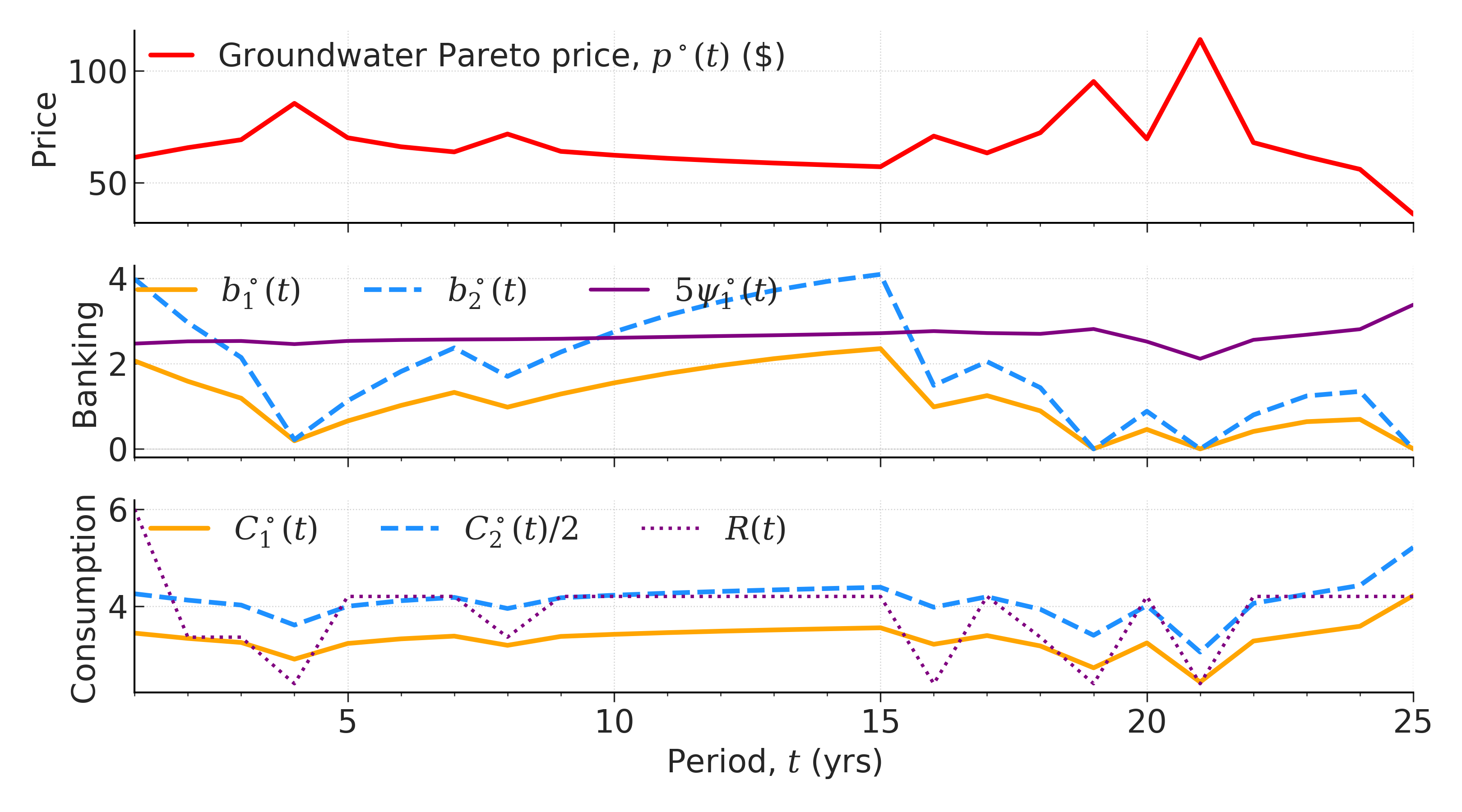}
\includegraphics[width=0.22\textwidth]{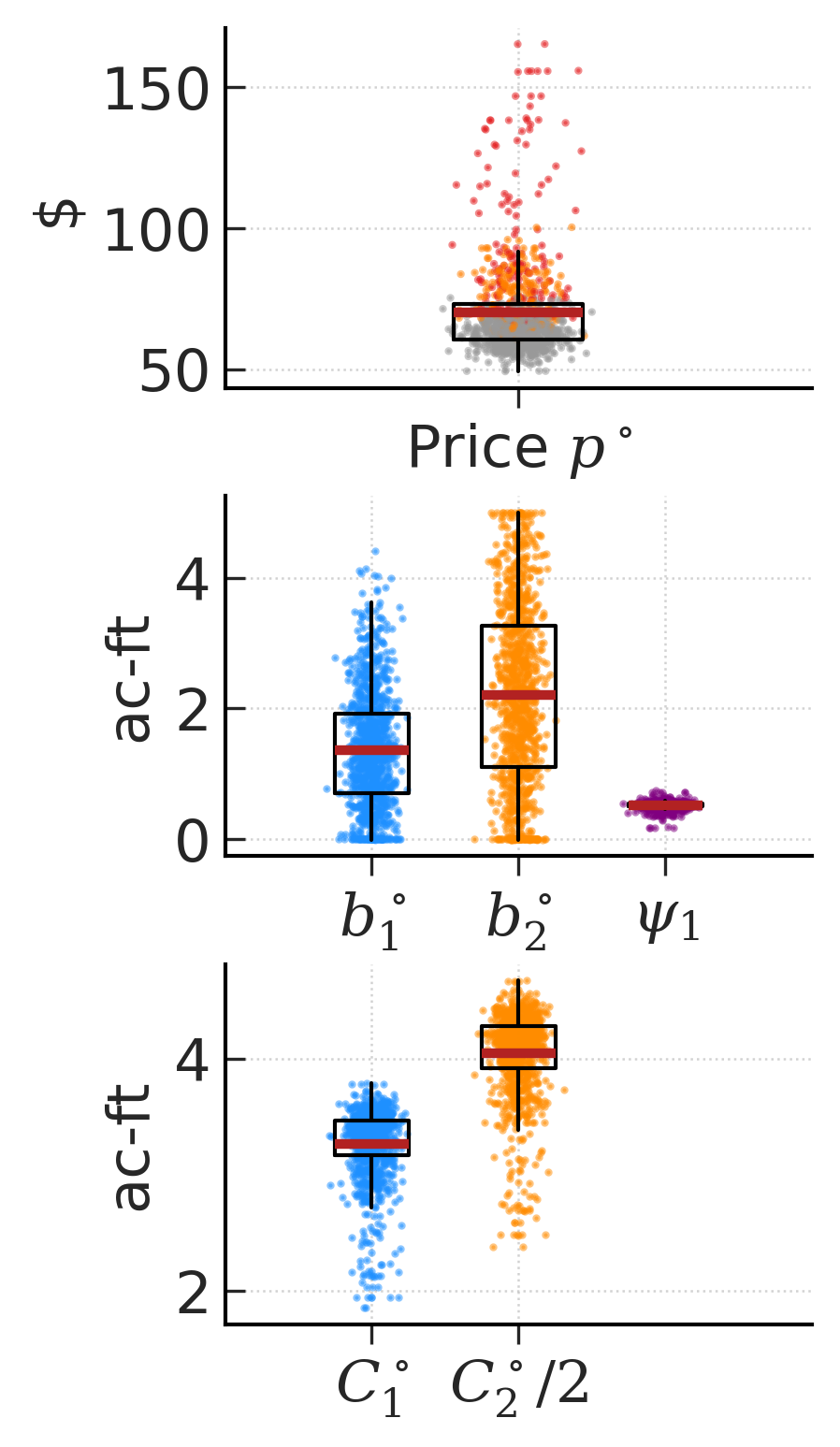}
\caption{Alfalfa vs Olive/Almond farming in the San Joaquin basin in California. All values are in ac-ft per acre of land; prices are per ac-ft and color-coded by $R(t)$ state. Because Farmer 2 has twice the acreage of Farmer 1, the panels show $C^\circ_2(t)/2$.}\label{fig:sj}
\end{figure}

Figure \ref{fig:sj} shows a typical trajectory over 25 years in the resulting San Joaquin groundwater market and the corresponding steady-state distributions of the equilibrium quantities. The obtained equilibrium is broadly similar to the case study of Section \ref{sec:statics}. Banking allows to maintain nearly constant consumption unless there are multiple consecutive or closely spaced droughts that fully deplete reserves. In San Joaquin the cap on $b^\circ_j(t) \le 5$ and the high variability of annual allocation (which vary by nearly a factor of 2) imply that reserves are at best sufficient for 2 years. They get fully depleted at $t=4$ (after 2 medium years and a drought) and again at $t=19$ (after 2 droughts and a medium year out of 4 years). 

The obtained groundwater equilibrium prices in the range $p^\circ(t) \in [50,150]$ per ac-ft are in the ballpark of likely transactions; presently California GSAs sell water for $\$10-30/$ac-ft and penalize overpumping at $\$200-500$/ac-ft which serve as lower and upper bounds for economically reasonable $p(t)$. As expected, given the relative high revenues from olives and almonds, which bring in nearly triple the cash than alfalfa and also come with lower irrigation costs, on the margin the more productive  orchards are prioritized ahead of forage. We find that 0.45-0.55 ac-ft of rights are typically traded from the alfalfa farmer to the orchard farm; the latter ends up consuming about 105\% of her nominal needs to collect extra revenue, while the lower-value alfalfa farm only uses on average 61\% of nominal needs, preferring to sell groundwater rights. In particular $C^\circ_2 > C^\circ_1$ even though both olives and almonds need nominally less water than alfalfa.

\section{Conclusion}

The presented model offers a first pass at modeling dynamic groundwater markets, serving as a foundation for further development. Multiple future directions ought to be explored and studied. To begin with, new algorithms are needed to compute (sub-game perfect) Nash equilibria in the presence of multiple agents. While our current approach can be straightforwardly extended to $N>2$ agents, it suffers from two scalability challenges. First, we do not have any guarantee on the best-response iteration approach, and hence no guarantee that an equilibrium can be approximated. Second, we directly grid out the input space $\bfw \in \cX^N$ which grows exponentially in $N$. To resolve this challenges therefore requires more sophisticated t\^{a}tonnement methods, as well as alternatives for training and fitting $\widehat{V}_j$'s. Ideas from Multi-Agent Reinforcement Learning should prove fruitful for both. In tandem with above, with multiple agents we expect that there are multiple NEs and the question of equilibrium refinement and selection looms large. This calls for theoretical investigation of how to characterize the range of feasible equilibria and criteria for selecting among them, and computational approach to find a specific desired NE.

The presented model postulates a Rational Expectations Equilibrium, meaning that all agents are fully rational and have complete knowledge of the system dynamics and current state, such as knowing the banking and consumption of all other agents. It would be worthwhile to compare to setups with bounded rationality, for example where agents only base their decisions on the aggregate pumping or other summary statistics, rather than the entire profiles $\bf{b}, \bf{w}$.

Next, in the presented paradigm, the regulator acts as the invisible hand, merely enforcing market clearing and the use of the Pareto price $p^\circ$. More generally, the regulator may have objectives of their own, linked for instance to collective equity among agents, or to environmental protection (i.e., accounting for further externalities that we have not mentioned yet). Such active groundwater management can be captured through a Leader-Follower game: the regulator as the leader dynamically optimizes her criterion, while anticipating the strategic collective response from the agents, that in turn competitively compete with each other. This would add a further layer of optimization to our model. 

Last but not least, there is much empirical work to be done, from more careful calibration of the model to realistic groundwater basins, to studying a range of alternative market mechanisms, for instance use of auctions rather than clearing prices, or the use of non-Pareto prices that requires pro-rata rationing among the unbalanced buyers and sellers. 

\subsubsection*{Acknowledgments} 
IC  acknowledges support from the National Science Foundation grant DMS-2407549. ML acknowledges support from the National Science Foundation grant DMS-2407550. We thank John Inston for the recharge calibrations in Section \ref{sec:calilbration}.

\subsubsection*{Competing Interests} 
The authors declare no competing interests. 

\subsubsection*{Data Availability}
The developed accompanying Python code used for numerical computations is available from the corresponding author on reasonable request.


\newcommand{\etalchar}[1]{$^{#1}$}

\appendix
\section{Proofs}\label{sec:apend1}

\subsection{Proof of Theorem \ref{th:main}}
Recall that\footnote{Here, to simplify notation we put $\bfpi_j = \bfpi_j(t)$, and $p=p(t)$.} $(\bfpi_{-j}, p, \bfw)$  belongs to some compact set in the Euclidean space $\bR^{2J+1}$, and $r$ is discrete. Let us first prove by backward induction that $\cV_j(\bfpi_{-j}, p; t, \bfw, r)$ is continuous in all its variables. Clearly this is true for $t=T$. Fix the time $t+1$ optimal  strategy $(\widehat \bfpi, \widehat p):=( \bfpi^\circ(t+1), p^\circ(t+1))$. By continuity of  $G_j, U_j$, the kernels \eqref{eq:trProb}, the induction hypothesis and the transition dynamics \eqref{eq:dynWj}, the function 
\begin{align*}
\cQ_j(\bfpi_{-j}, p, & \ t, \bfw, r;  C_j,\psi_j)  =  U_j \bigl( G_j(t, C_j) + \psi_j p\bigr)  \\
 &+ e^{-\delta} \bE^{\bfpi}[ \cV_j(\widehat \bfpi_{-j}, \widehat p; t+1, \bfw + \bfR(t+1) - \bfC - \boldsymbol{\psi}, R(t+1) ) \mid \bfW(t) = \bfw, R(t) = r] 
\end{align*}
is continuous in all its variables $\bfpi_{-j}, p, \bfw, r, C_j,\psi_j$. We also note that the action space $\cA_j(\bfw)$ does not depend on $\bfpi_{-j},p,r$\footnote{though the quoted result covers the case where the action space is state-dependent.}   it is upper and lower hemicontinous in $\bfw$, compact, and nonempty. Thus, by a version of Berge Maximum Theorem \cite[Theorem 17.31]{AliprantisBorderBook}, the maximizer in \eqref{eq:DPP-NE} is attained, $\cV_j(\bfpi_{-j}, p; t, \bfw, r)$ is well defined, and it is continuous in $\bfpi_{-j}, p, \bfw, r$.  Moreover, for all $\bfpi_{-j}, p, \bfw, r$, the set of maximizers, or the set of optimal best response strategies, $\cR_j(\bfpi_{-j}, p; t, \bfw, r)$,  is nonempty, compact and the corresponding multi-function 
$$
(\bfpi_{-j}, p, \bfw, r) \rightrightarrows \cR_j(\bfpi_{-j}, p; t, \bfw, r) := \argmax_{C_j,\psi_j} \cQ_j(\bfpi_{-j}, p, t, \bfw, r;  \ C_j,\psi_j), 
$$
is upper-hemicontinuous. 

By concavity of $G_j, U_j,$ monotonicity of $U_j$, for a fixed $p>0$, the function $(C_j,\psi_j) \mapsto U_j(G_j(t,C_j) + \psi_j p)$ is concave. For fixed $\widehat \bfpi_{-j}, \widehat p, t$, 
since the mapping $(C_j,\psi_j)\mapsto x' = (\bfw + \bfR(t+1) - \bfC - \boldsymbol{\psi}, R(t+1))$ is affine, and $x'\mapsto\cV_{j}(\widehat \bfpi_{-j}, \widehat p; t+1, x')$ is continuous by induction hypothesis, we get that $(C_j,\psi_j)\mapsto \cV_{j}(\widehat \bfpi_{-j}, \widehat p; t+1, x')$ is concave, and since expectation preserves concavity,  we conclude that for fixed $\bfpi_{-j}, p, t$, the objective function $\cQ$ is concave in $(C_j,\psi_j)$. Therefore, using this and since the feasible set $\cA_j(\bfw)$ is convex, it follows that $\cR_j(\bfpi_{-j}, p; t, \bfw, r)$ is also convex-valued. 

For every $p$, consider the joint best response mapping 
\[
\cR(\bfpi, p; t, \bfw, r) = \Pi_{j=1}^J \cR_{j}(\bfpi_{-j}, p ; t, \bfw, r).
\]
Since each $\mathcal R_j(\boldsymbol\pi_{-j},p;t,\mathbf w,r)$ is nonempty, compact, convex-valued and upper hemicontinuous in $\boldsymbol\pi_{-j}$, so is the joint best-response correspondence $\cR(\bfpi, p; t, \bfw, r)$ in $\bfpi$. By Kakutani fixed point theorem, there exists $\boldsymbol\pi^\circ$ such that
\begin{equation}\label{eq:pi-circ}
    \boldsymbol\pi^\circ(t, \bfw,r) \in \mathcal R(\boldsymbol\pi^\circ,p;t,\mathbf w,r).
\end{equation}

Given $\bfpi^\circ(t,\bfw, r)$, let us consider the aggregate demand correspondence 
\[
\Psi(p; \bfw, r) = \Set{\sum_{j=1}^J \psi_j \ : \ (\psi_1,\ldots,\psi_J) \in \cR(\bfpi^\circ, p; t, \bfw, r)}. 
\]
Since $\cR(\bfpi^\circ, p; t,\bfw, r)$ is compact, convex, and upper-hemicontinous, and the mapping $\sum_j x_j$ is continuous, the set-valued mapping $\Psi$ is nonempty, compact-valued, convex, and upper-hemicontinous in $p$; cf. \cite[Theorem~17.23]{AliprantisBorderBook}.   Therefore, employing the Berge maximum theorem, the minimizer exists, and the function 
\[
g(p,\bfw, r)=\min_{\psi\in\Psi(p,\bfw, r)} |\psi|
\]
is continuous in $p$. Consequently, by compactness of $[\underline{p},\overline{p}]$, there exists
\begin{equation}\label{eq:pcirc}
p^\circ(t, \bfw,r)\in \argmin_{p\in[\underline{p},\overline{p}]} g(p,\bfw, r).
\end{equation}
Next, for $p^\circ(t; \bfw, r)$, since $\cR_j(\boldsymbol\pi_{-j}^\circ,p^\circ;t,\mathbf w,r)$ is compact-valued and has a measurable graph, by \cite[Theorem~18.13]{AliprantisBorderBook}, there exists a measurable selector 
$$
(C_j^\circ(t;\bfw, r), \psi_j^\circ(t; \bfw, r)) \in \cR_j(\bfpi_{-j}^\circ,p^\circ(t; \bfw, r)),
$$
and since $\Psi(p; \bfw,r)$ is compact-valued and the function $\psi \mapsto |\psi|$ is continuous, a measurable version of Berge’s maximum theorem implies the existence of a measurable selector such that 
\[
\sum_{j=1}^J \psi_j^\circ(t,\bfw, r) \in \argmin_{\psi\in\Psi(p^\circ(\bfw, r),\bfw, r)} |\psi|.
\]

The final verification step follows at once. The strategy $(C_j^\circ(t;\bfw, r), \psi_j^\circ(t; \bfw, r))$ is optimal, and $p^\circ(t,\bfw,r)$ minimizes the market clearing conditions. Thus, this is a sub-game perfect equilibrium. Continuity of the value function $\cV_j(\bfpi_{-j}, p^\circ; t, \bfw,r)$, and $\cV_{J+1}(\bfpi; t, \bfw,r)$ follows by Berge maximum theorem.

\subsection{Proof of Theorem~\ref{th:ne-market-clearing}}

Let us denote the continuation term in $\cQ_j$, as function of $\psi_j$, as 
\begin{equation}\label{eq:Z-contTerm}
\psi_j\mapsto Z_j(t,\psi_j) = e^{-\delta} \bE^{\bfpi}\!\left[ \cV_j(\widehat{\bfpi}_{-j},\widehat p; t+1, \bfw+\bfR(t+1)-\bfC-\boldsymbol\psi, R(t+1)) \ \mid \bfW(t) = \bfw, \ R(t)=r\right]. 
\end{equation}
To prove the result, we need to study how the continuation value depends on $\psi_j$. While $\cV_j$ might not be differentiable, recall that $Z_j$ is concave, and thus its superderivatives exist on the feasible convex set.  

We fix $(t,\bfw,r)$, and if there is no confusion, we drop writing these variables, e.g., we write $\cR_j(\bfpi_{-j}, p)$ instead of $\cR_j(\bfpi_{-j}, p; t, \bfw,r)$. 
		
		Next we prove that for small enough prices $p$, for any $(C_j,\psi_j)\in \cR_j(\bfpi_{-j},p)$ we have that $\psi_j<0$. For a fixed $C_j$, and $p>0$, define
		\[
		H_j(\psi_j) = U_j(G_j(C_j) + p \psi_j) + Z_j(\psi_j),
		\] 
		where $Z_j$ is the continuation term \eqref{eq:Z-contTerm}.  By concavity of $U_j$ and $Z_j$, we have that $H_j$ is concave with superdifferential 
		\[
		\partial^+ H_j(\psi_j) = U_j'(G_j(C_j) + p\psi_j) p + \partial^+ Z_j(\psi_j). 
		\]
		This means that for any supergradient $\xi\in \partial^+ H_j(\psi_j)$, there exists $\zeta\in\partial^+ Z_j(\psi_j)$, such that 
		\[
		\xi = U_j'(G_j(C_j) + p\psi_j) p + \zeta.  
		\]
	By assumption that $U'\leq M_U$, we get	that $\xi \leq M_U p + \zeta$,  and by Lemma~\ref{lemma:supergrad1}, for any $p <\underline{p}=M_G$ we have that every supergradient $\xi <0$, and since $H_j$ is concave, we have that $H_j$ is strictly decreasing in $\psi_j$, for $p$ small enough.  Thus, its largest value is attained for the smallest $\psi_j$ that meets the constraint $-\sum_{i\neq j} w_j \leq C_j + \psi_j$, which is always negative $\psi_j <0$.  
	
	Similarly, $\xi \geq M_U p - M_UM_G$ and for any $p>\overline{p}=1/M_G$, we have that $\xi > 0$, and thus $H_j$ is increasing, with maximum attained at $W_j-C_j$. Thus, for large $p>\overline{p}$ and all $C_j < W_j\wedge \overline{c}_j$, the farmer will prefer to sell $\psi_j>0$.

	Next, for every $j$, consider the set-valued mapping 
		\[
		p\rightrightarrows \Psi_j(p):= \set{\psi_j \, :\, \exists C_j, \ (C_j,\psi_j) \ \in \cR_j(\bfpi^\circ_{-j},p; t, \bfw,r)}. 
		\]
From the above, we conclude 
		\[
		\sup \Psi_j(\underline{p}) < 0 < \inf \Psi_j(\overline{p}).
		\]
	Hence, the aggregate  demand $\Psi(p)$ satisfies 
		\[
		\sup \Psi(p^1) < 0 <  \inf \Psi(p^2).
		\]
		Consequently, there exists $p^\circ$ such that $0\in \Psi(p^\circ)$, and this is the minimizer of \eqref{eq:pcirc}.

\begin{lemma}\label{lemma:supergrad1}
Every supergradient $\zeta\in\partial ^+ Z_j(t,\psi_j)$, satisfies
\begin{equation}\label{eq:zeta3}
- M_U M_G e^{-\delta}\le \zeta \le 0. 
\end{equation}
\end{lemma}
\begin{proof} 
The proof follows by standard arguments and a backward induction procedure. For the sake of completeness, we present it here. 

Each induction step consist of two substeps. First, prove that  
\[
w_j\mapsto \cV_j ({\bfpi}_{-j}, p; t, \bfw, r )
\]
is concave (already proved) and any supergradient $ \xi\in \partial_w^+\cV_j({\bfpi}_{-j}, p; t, \bfw, r )$ satisfies
\begin{equation}\label{eq:xi4}
0\leq \xi \leq M_U M_G.
\end{equation}
Second, prove that the map $\psi_j \mapsto Z_j(t,\psi_j)$ is concave and every supergradient $\zeta\in \partial_{\psi_j}^+ Z_j(t,\psi_j)$ satisfies \eqref{eq:zeta3}. 

For $t=T$, the mapping $w_j\mapsto \cV_j ({\bfpi}_{-j}, p; T, \bfw, r ) = U_j(G_j(T,w\wedge \overline{c}_j))$
is concave (since both $U_j, G_j(t,\cdot)$ are concave in increasing), and by chain rule 
$\partial_{w_j} \cV_j = U_j'(G_j(T, w_j))\partial_C G_j(T,w_j)$. Thus,  
\[
0\leq \partial_{w_j}\cV_j \leq M_U M_G, 
\]
and \eqref{eq:xi4} is satisfied. Note that $Z_j(T,\psi_j):=0$, and \eqref{eq:zeta3} holds true. 

Assume the assertion holds true for $t+1$, that is $\bar w_j \mapsto \cV_j(\widehat{\bfpi}_{-j}, \widehat{p}; t+1, \bar \bfw, R(t+1))$  is concave and every its supergradient in $\bar w_j$ belongs to $[0, M_U M_G]$. Take the affine mapping  $\bar \bfw= \bfw + \bfR(t+1) - \bfC - \boldsymbol{\psi}$. Then, since composition with affine transformation preserves concavity, we have that 
\[
(C_j,\psi_j,w_j) \mapsto \cV_j(\widehat{\bfpi}_{-j}, \widehat{p}; t+1, \bfw + \bfR(t+1) - \bfC - \boldsymbol{\psi} )
\]
is jointly concave and its every supergradient in $w_j$ belongs to $[0, M_U M_G]$. 
Since the conditional expectation preserves concavity 
\[
(C_j,\psi_j,w_j)  \mapsto e^{-\delta} \bE^{\bfpi}_t[\cV_j(\widehat{\bfpi}_{-j}, \widehat{p}; t+1, \bfw + \bfR(t+1) - \bfC - \boldsymbol{\psi} )]
\]
is also jointly concave. 
Consequently, the function 
\[
(C_j,\psi_j,w_j)  \mapsto \cQ_j(\bfw, C_j, \psi_j) =  U_j(G_j(t,C_j) + p\psi_j) + e^{-\delta} \bE^{\bfpi}_t[ \cV_j(\widehat{\bfpi}_{-j}, \widehat{p}; t+1, \bfw + \bfR(t+1) - \bfC - \boldsymbol{\psi} )]  
\]
is also jointly concave. Since the first term  $U_j(G_j(t,C_j) + p\psi_j)$ does not depend on $w_j$, and the second is affine in $w_j$ with coefficient 1, every suppergradient of  $w_j\mapsto \cQ_j(\bfw, C_j, \psi_j)$ belongs to $[0, M_UM_G]$. 

Since $\cQ_j$ is jointly concave in $(w_j, C_j, \psi_j)$ and the set $(C_j,\psi_j,\cA_j(\bfw)$ is convex,  in view of \cite[Section 3.2.5.]{BoydVandenberghe2004Book} $w_j \mapsto \cV_j(\bfpi_{-j}, p; t, \bfw, r)$ is concave. 
Let $\bfw_1,\bfw_2$ fixed and $(C_j^*,\psi_j^*)\in \cA_j(\bfw_1)$ be the optimizer of $\cV_j(\bfpi_{-j}, p; t, \bfw_1,r)$. Then 
\[
\cV_j(\bfw_2,r) - \cV_j(\bfw_1,r) \leq Q_j(\bfw_2, C_j^*, \psi_j^*) - \cQ_j(\bfw_1, C_j^*, \psi_j^*). 
\]
Since supergradients in $w_j$ of $Q_j$ are bounded by $M_UM_G$, we have 
\[
Q_j(\bfw_2, C_j^*, \psi_j^*) - \cQ_j(\bfw_1, C_j^*, \psi_j^*) \leq M_UM_G (w_{j,2} - w_{j,1}),
\] 
which combined with the previous inequality, we conclude that any $\xi\in\partial_{w_j}^+\cV_j(\bfw,r)$ is bounded from above by $M_UM_G$. Non-negativity of $\xi$ follows from the monotonicity of $\cV_j(\bfw,r)$ in $w_j$. Thus, for any $\xi\in\partial_{w_1}^+ \cV_j(\bfw,r)$, we have \eqref{eq:xi4}. 

Noting that $\bar w_j = w_j + R(t)- C_j -\psi_j$ is affine in $\psi_j$, with slope -1, and that the function  
\[
Z_j(t-1,\psi_j) =  e^{-\delta} \bE^{\bfpi}_{t-1} \left[ \cV_j(\widehat{\bfpi}_{-j},\widehat p; t, \bfw+\bfR(t)-\bfC-\boldsymbol\psi, R(t)) \right].  
\]
is concave, and since the supergradient is computed as 
\[
\zeta = e^{-\delta} \eta, \quad \eta\in\partial_{w_j}^+ \cV_j(\widehat{\bfpi}_{-j},\widehat p; t, \bfw), 
\]
we deduce that any supergradient of $\zeta\in\partial_{\psi}^+Z_j(t-1,\psi_j)$ satisfies \eqref{eq:zeta3}.

\end{proof}

\end{document}